\documentclass[11pt]{article}
\usepackage{fullpage}
\usepackage{graphicx}
\usepackage{amsmath}
\usepackage{amsfonts}
\usepackage{amssymb}
\usepackage{float}
\usepackage{mathrsfs}
\usepackage[margin=1in]{geometry}
\usepackage{bbm} 
\usepackage{soul}
\usepackage[page]{appendix}
\usepackage[round]{natbib}
\usepackage{amsthm}
\usepackage{tikz}
\usetikzlibrary{positioning, arrows.meta}
\usepackage{framed}
\usepackage{mathtools}
\usetikzlibrary{positioning}

\newcommand{\ind}{\perp\!\!\!\!\perp} 

\newcommand{\keywords}[1]
{
  \small	
  \textbf{\textit{Keywords--}} #1
}

\newtheorem{theorem}{Theorem}
\newtheorem{assumption}{Assumption}
\newtheorem{proposition}{Proposition}
\newtheorem{corollary}{Corollary}

\newtheorem{remark}{Remark}

\title{\bf A Sensitivity Analysis Framework for Causal Inference Under Interference}
\date{First Version: November 26, 2025; Current Version:  June 30, 2026}
\author{Matvey Ortyashov  \qquad AmirEmad Ghassami \\ Department of Mathematics and Statistics, Boston University}
\begin{document}
\maketitle

\begin{abstract}
   In many applications of causal inference, the treatment received by one unit may influence the outcome of another, a phenomenon referred to as interference. Although there are several frameworks for conducting causal inference in the presence of interference, practitioners often lack the data necessary to adjust for its effects. In this paper, we propose a weighting-based sensitivity analysis framework that can be used to assess the systematic bias arising from ignoring interference. Unlike most of the existing literature, we allow for the presence of unmeasured confounding, and show that the combination of interference and unmeasured confounding is a notable challenge to causal inference. We also study a third factor contributing to systematic bias: lack of transportability. Our framework enables practitioners to assess the impact of these three issues simultaneously through several easily interpretable sensitivity parameters that can reflect a wide range of intuitions about the data.
\end{abstract}

\section{Introduction}\label{secintroduction}
Researchers in causal inference frequently invoke the assumption of no interference, which states that the treatment received by any unit should only affect the outcome of that unit \citep{Rubin1980}. Due to the widespread adoption of this assumption, most literature ignores the presence of \textit{spillover effects}, which occur when one unit's treatment may ``spill over," and have a causal impact on (i.e., interfere with) another's outcome. Ignoring spillover effects can result in systematic bias, a form of bias that cannot be mitigated by increasing the sample size. Concerns about the bias introduced by spillover effects are particularly relevant in fields where interference is common, such as the study of infectious diseases \citep{Vanderweele2015}. 

While the role of ignored interference in creating systematic bias is widely recognized, few authors have explicitly explored the structure, the components, and the extent of this bias. As a result, notwithstanding a proliferation of recent research concerned with proper estimation strategies in the presence of observed interference \citep{Liu2016, Forastiere2021, Lee2023, McNealis2024, Tortú2024, Papadogeorgou2023, Forastiere2024}, the subject of sensitivity analysis for ignored interference has hitherto been underexplored. Despite the limited amount of literature concerning the topic, sensitivity analyses can play an important role in the presence of spillover effects. In many settings, practitioners may incorrectly assume that there are no spillover effects within a given set of data or, due to cost or privacy concerns, may lack the ability to learn the structure of the interference present in the data. A sensitivity analysis framework would enable researchers facing these issues to evaluate how ignoring interference impacts their estimates and gauge the strength of interference necessary to change their findings. 

In this paper, we study the structure of the bias that arises due to failing to adjust for various complexities introduced by the presence of interference. Unlike the majority of the existing literature on interference (with some exceptions, such as \cite{Chen2025}, \cite{Wu2025}, and \cite{Khot2025}), we allow for the presence of unobserved confounders. Importantly, we will demonstrate that if interference is present, the impacts of unmeasured confounding on the bias differ from those that exist without interference, making the combination of unmeasured confounding and ignored interference a distinct challenge to causal inference. Using our bias decomposition, we propose a sensitivity analysis framework that parametrizes the bias through several easily interpretable sensitivity parameters that reflect a practitioner's beliefs about factors such as the strength of interference effects and unmeasured confounding. Our analyses also account for another data complexity that can occur alongside interference: the presence of undefined potential outcomes. 

Furthermore, we consider settings where the researcher aims to \emph{transport} their causal findings from a reference domain to a target domain. For this case, we investigate a third source of bias: lack of transportability. Causal transportability, alongside several concepts closely related to it, such as external validity, generalizability, data fusion, and transfer learning, has garnered significant attention over the past several years \citep{Bareinboim2016, Mitra2022, Huang2024, Degtiar2023, Colnet2024,Vuong2025}. As discussed by \cite{Buchanan2023} as well as \cite{Bhadra2025}, though there is a wide range of considerations about transporting causal effects in the presence of interference (one of which is that transportability requires similar patterns of interference between populations, as noted by \cite{Hernán2011}), few authors have explicitly examined them. We study how our causal estimands change between different populations, and extend our sensitivity analysis to account for lack of transportability.  

To the best of our knowledge, our paper is the first to develop a comprehensive sensitivity analysis framework that allows practitioners to simultaneously account for the bias that comes from ignored interference, unmeasured confounding, and lack of transportability. Although \cite{Forastiere2021} discuss the bias of an estimator that does not adjust for interference and unmeasured confounding, they do not present sensitivity parameters or offer a structured approach that can be employed by a practitioner wishing to calculate the bias of their naive estimator. To date, only \cite{Vanderweele2015} have developed a formal sensitivity analysis framework for unmeasured confounding in the presence of interference, though their methodology does not account for ignored interference or lack of transportability. Unlike their work, our method not only incorporates the aforementioned data complexities, but also does not require strong assumptions about the outcome generating process, the distribution of the unmeasured confounder, or the selection bias. However, if one is willing to make certain assumptions, we demonstrate that our bias decomposition can be significantly simplified. Our bias decomposition generalizes previous weighting-based results, and is related to the approaches of \cite{Shen2011}, \cite{Hong2021}, and \cite{Huang2024}.

The rest of this paper is organized as follows. Section \ref{problem} describes the problem setting. In Section \ref{sensitivity}, we provide the main theoretical results for bias arising from unmeasured confounding and interference. Section \ref{sectiontransportability} extends the results of Section \ref{sensitivity} to the task of causal transport. Section \ref{undefined} extends the results of Sections \ref{sensitivity} and \ref{sectiontransportability} to undefined potential outcomes. We conclude in Section \ref{sectionconclusion}. All of the proofs are provided in the Appendix.

\section{Problem Description}\label{problem}
Consider a setting in which unit $i \in \mathcal{I}$ receives treatment $A_i \in \{0,1\}$ and has an observed outcome $Y_i$. Without interference, unit $i$'s outcome depends only on unit $i$'s treatment; however, in the presence of interference, $Y_i$ may also depend on the treatments received by other units \citep{Rubin1974, Rubin1990}. We refer to those units whose treatments affect $Y_i$ as the \textit{neighbors} of unit $i$. We let $\mathcal{N}_i$ denote the set of indices of unit $i$'s neighbors, and let the vector $\mathbf{A}_{\mathcal{N}_i}$ represent their treatments. We do not impose any restrictions on which units can belong to $\mathcal{N}_i$. Thus, our approach allows for more flexibility than the common setting of partial interference, in which units are partitioned into blocks, and interference can only occur within blocks, not across them \citep{Sobel2006}. We adopt the potential outcomes framework, where units' outcomes under different realizations of the treatments are considered distinct random variables \citep{Rubin1974}. Specifically, if $A_i$ is set to $a$ and $\mathbf{A}_{\mathcal{N}_i}$ is set to $\mathfrak{a}$, we write the potential outcome of unit $i$ as $Y_i^{(a, \mathfrak{a})}$. Furthermore, we let  $\mathbf{X}_i$ and $\mathbf{U}_i$ denote the observed and unobserved pre-treatment covariates respectively. For the context of causal transport, we use the variable $S_i$ to demarcate units in the reference population ($S_i = 1$) from those in the target population ($S_i = 2$). We assume that our observational data only contains units with $S_i = 1$, and that the two populations may differ in terms of covariates and treatment generating mechanisms. 

Under neighborhood interference, a unit with $n_i$ neighbors may have $2^{n_i+1}$ distinct potential outcomes. Nevertheless, in a majority of applications, many of these potential outcomes would equal each other. Thus, we assume that there exists a scalar-valued function $g_i(\mathbf{A}_{\mathcal{N}_i}) = G_i$ which acts as a summary of the treatments received by unit $i$'s neighbors and reduces the number of potential outcomes we need to consider. This assumption is formalized below.
\begin{assumption}[Outcome Exposure Mapping]\label{neighint}
For all $i \in \mathcal{I}$, there exists a function $g_i: \{0,1\}^{n_i} \to \mathbb{R}$ such that for any two different sets of treatments received by a unit's neighbors, $\mathfrak{a}_1$ and $\mathfrak{a}_2$, we have that
\begin{equation*}
    Y_i^{(a, \mathfrak{a}_1)} = Y_i^{(a, \mathfrak{a}_2)}
\end{equation*}
for $a \in \{0,1\}$, as long as $g_i(\mathfrak{a}_1) = g_i(\mathfrak{a}_2)$. We assume that $g_i(\mathbf{A}_{\mathcal{N}_i})$ is discrete, and that, without loss of generality, for all $i \in \mathcal{I}$, $g_i(\mathbf{A}_{\mathcal{N}_i}) \in \{0, 1, ...,  g_{\text{max}} \}$. Moreover, we assume that $g_i$ only depends on $i$ through the dimension of the domain, and that $ g_{\text{max}} < \max_i 2^{n_i}$.
\end{assumption}
Effectively, Assumption \ref{neighint} states that if $A_i = a$ and $G _i= g$, the potential outcome can be written as $Y_i^{(a, g)}$. The function $g_i$ is a particular case of the \textit{exposure mapping} described by \cite{Aronow2017}. We call function $g_i$ the \textit{outcome exposure mapping}, $G_i$ the \textit{neighborhood treatment}, and $A_i$ the \textit{personal treatment} of unit $i$.\footnote{The exposure mapping assumption is not the only way to define potential outcomes in the presence of interference: alternative definitions, typically based on a given treatment allocation strategy, have also been posited in the literature \citep{TchetgenTchetgen2012, Liu2016, Lee2023}.} Throughout the rest of this paper, we drop the subscript $i$ and refer to generic instances of the aforementioned variables without the subscript. 

Next, we propose a version of the consistency assumption, which relates the unit's observed outcome, $Y$, to its potential outcomes.
\begin{assumption}[Consistency under Interference] \label{consistency}
The potential outcome $Y^{(a)}$ satisfies $Y^{(a)} = Y^{(a,G)}= \sum_{g = 0}^{ g_{\text{max}}} \mathbb{I}(G=g) \cdot  Y^{(a,g)}$. Moreover, the observed outcome $Y$ satisfies $Y =  \sum_{a = 0}^{1} \mathbb{I}(A=a) \cdot  Y^{(a)}$.
\end{assumption}

Taken together, the two components of Assumption \ref{consistency} imply that if a unit receives personal treatment $a$ and neighborhood treatment $g$, its outcome equals $Y^{(a, g)}$, i.e., $Y = \sum_{a =0}^1\sum_{g = 0}^{ g_{\text{max}}} \mathbb{I}(A = a, G=g) \cdot Y^{(a,g)}$. In addition to Assumptions \ref{neighint} and \ref{consistency}, we impose one ancillary constraint to develop tractable bias decompositions in Sections \ref{sensitivity} and \ref{sectiontransportability}. 
\begin{assumption}[Well-Defined Potential Outcomes] \label{assumptionwelldefined} The potential outcomes $Y^{(a)}$ and $Y^{(a,g)}$ for $a \in \{0,1\}$, $g \in \{0, 1, ..., g_{\text{max}} \}$ are well-defined for every unit.
\end{assumption}
In general, Assumption \ref{assumptionwelldefined} does not necessarily hold in the presence of interference, and thus, is not required for our bias decomposition. Nevertheless, we introduce  Assumption \ref{assumptionwelldefined} in order to avoid unnecessary complexities in notation that do not employ any new methodological tools and do not offer any additional insights into our main results. In Section \ref{undefined}, we explore cases in which Assumption \ref{assumptionwelldefined} may be violated, and extend our sensitivity analysis to such settings.

\subsection{Estimands}\label{estimands}
In this paper, we have two primary parameters of interest: the \textit{natural average main effect}, denoted by $\phi_{1}$, and the \textit{transported natural average main effect}, denoted by $\phi_{2}$. In order to describe these parameters, we begin by considering the \textit{controlled individual main effect}, $\tau(g)$:
\begin{equation*}
\tau(g) := Y^{(1,g)} - Y^{(0,g)} .
\end{equation*}
 $\tau(g)$ represents the individual causal effect of changing $A$ when $G$ is set to $g$, and can be viewed as a direct controlled effect \citep{Robins1992, VanderWeele2011Direct}. In addition to $\tau(g)$, for each unit, we can define the \textit{natural individual main effect}, $\kappa$, which measures the causal impact of changing $A$ while $G$ equals the actual value of neighborhood treatment received by the unit:
 \begin{equation*}
\kappa := Y^{(1, G)}-Y^{(0, G)} =Y^{(1)}-Y^{(0)}  = \sum_{g = 0}^{ g_{\text{max}}} \mathbb{I}(G=g) \cdot \tau(g).
\end{equation*}
We can then define both of our parameters of interest as averages of $\kappa$, one taken over the reference population ($S=1$), the other over the target population ($S=2$):
\begin{equation*}
\phi_{s} := \mathbb{E}\big[\underbrace{Y^{(1, G)}-Y^{(0, G)}}_{\kappa} |S=s\big] = \mathbb{E}\big[Y^{(1)}-Y^{(0)} |S=s\big]  = \sum_{g = 0}^{ g_{\text{max}}} \mathbb{E}\big[\mathbb{I}(G=g) \cdot \tau(g) |S=s\big].
\end{equation*}
As the average difference between $Y^{(1,G)}$ and $Y^{(0, G)}$, $\phi_{s}$ can be interpreted as measuring the impact of changing $A$ from 0 to 1 while $G$ is allowed to vary as it naturally does in the population where $S=s$.

\begin{remark}
Though we have hitherto discussed two parameters that measure the causal effect of changing the personal treatment $A$, we can also define $\textbf{spillover}$ \textbf{effects}, which measure the causal impact of changing the neighborhood treatment $G$. Similarly to the natural individual main effect, we can define the \textbf{natural individual spillover effect} at $A=a$, denoted by $\gamma(a)$:
\begin{equation*}
\gamma(a) = Y^{(a,G)}-Y^{(a,0)} = Y^{(a)}-Y^{(a,0)} =  \sum_{g=0}^{ g_{\text{max}}} \bigg\{\mathbb{I}(G=g) \cdot  Y^{(a,g)} \bigg\}- Y^{(a,0)}.
\end{equation*}
 $\gamma(a)$ measures the causal spillover effect for any given unit by comparing the potential outcome where $A=a$ and $G$ is set to the actual neighborhood treatment received by that unit with the potential outcome where $A=a$ and $G$ is set to 0. The \textbf{individual total effect} can be defined in terms of the individual natural main and spillover effects:
\begin{align*}
\text{Individual Total Effect} &= Y^{(1,G)} - Y^{(0,0)} = Y^{(1)} - Y^{(0,0)}  = \kappa + \gamma(0)
\end{align*}
While we introduce spillover effects in this section,  our focus in this work remains developing a sensitivity analysis framework for $\phi_{1}$ and $\phi_{2}$. Thus, we only discuss spillover effects insofar as they figure into our bias decompositions for the average natural main effects. 
\end{remark}

In settings where all of the relevant variables are observed, the requirements of conditional exchangeability and positivity are typically posited to allow for identification. However, since we wish to develop a framework that allows for the existence of unobserved confounders and lack of transportability, we will adjust these assumptions. Before doing so, we differentiate between several different categories of covariates in our setting. Namely, we say that $\{\mathbf{X}_{AY}, \mathbf{U}_{AY}\}$ contains all common causes of $A$ and $Y$, $\{\mathbf{X}_{GY}, \mathbf{U}_{GY}\}$ contains all common causes of $\mathbf{A}_{\mathcal{N}}$ and $Y$, $\{\mathbf{X}_{AG}, \mathbf{U}_{AG}\}$ contains all common causes of $A$ and $\mathbf{A}_{\mathcal{N}}$, $\{\mathbf{X}_{AS}, \mathbf{U}_{AS}\}$ contains all common causes of $A$ and $S$, $\{\mathbf{X}_{GS}, \mathbf{U}_{GS}\}$ contains all common causes of $\mathbf{A}_{\mathcal{N}}$ and $S$, and $\{\mathbf{X}_{SY}, \mathbf{U}_{SY}\}$ contains all common causes of $Y$ and $S$. Using this notation, we posit the following modified versions of positivity and conditional exchangeability under interference.

\begin{assumption}[Positivity]\label{positivity}
For $a \in \{0,1\}$ and $s \in \{1,2\}$ we have, almost surely in $( \mathbf{X}, \mathbf{U})$ and $(G, \mathbf{X}, \mathbf{U})$
\begin{equation*}
p(A=a|S = s,\mathbf{X}, \mathbf{U}) > 0  \text{ and } p(S=s|G,\mathbf{X}, \mathbf{U}) > 0.
\end{equation*}
\end{assumption}
\begin{assumption}[Weak Conditional Exchangeability]\label{exchangeability}
For $a \in \{0, 1\}$, we have
\begin{equation*}
Y^{(a)}  \ind A| S, \tilde{\mathbf{X}},\tilde{\mathbf{U}}
\end{equation*}
where $\tilde{\mathbf{X}} =  \mathbf{X}_{AY} \cup \mathbf{X}_{AG} \cup \mathbf{X}_{AS}$ and $\tilde{\mathbf{U}} =  \mathbf{U}_{AY} \cup \mathbf{U}_{AG}  \cup \mathbf{U}_{AS}$ reflect all covariates which are the causes of $A$. 
\end{assumption}
Because Assumption \ref{exchangeability} allows for the existence of a set of unobserved covariates that ensures independence, it is significantly weaker than the version of exchangeability typically presented in the absence of unmeasured confounding.
\begin{remark}
     Assumption \ref{exchangeability} highlights three additional challenges introduced by the combination of interference and lack of transportability. 
     \begin{itemize}
         \item Under interference, one needs to condition on common causes of $A$ and $\mathbf{A}_{\mathcal{N}}$, even if these variables are not common causes of $A$ and $Y$, and do not need to be adjusted for in the absence of interference.
         \item Since the reference and target populations may differ in terms of their treatment generating mechanisms, $S$ also serves as a common cause of $A$ and $\mathbf{A}_{\mathcal{N}}$, and thus, must be adjusted for.
         \item Because the reference and target populations may differ in terms of their covariate distributions (namely, causes of $A$ and $\mathbf{A}_{\mathcal{N}}$), there is a collider structure that necessitates adjusting for common causes of $A$ and $S$.
     \end{itemize}

\end{remark}

Under Assumptions \ref{neighint}-\ref{exchangeability}, one could identify $\phi_{1}$ if they were to observe $\tilde{\mathbf{U}}$. One possible approach is to use inverse probability weighting (IPW), as in Proposition \ref{unbiasedIPW}.
\begin{proposition}[Identification Formula for $\phi_{1}$]\label{unbiasedIPW}
Under Assumptions \ref{neighint} through \ref{exchangeability}, $\phi_{1}$ can be identified as
\begin{equation*}
    \phi_{1} = \mathbb{E} \bigg[ \frac{\mathbb{I}(A=1) \cdot Y}{p(A=1|S=1,\tilde{\mathbf{X}}, \tilde{\mathbf{U}})}  - \frac{\mathbb{I}(A=0) \cdot Y}{p(A=0|S=1, \tilde{\mathbf{X}}, \tilde{\mathbf{U}})} \bigg|S=1\bigg].
\end{equation*}
\end{proposition}
Assumptions \ref{neighint}-\ref{exchangeability} are sufficiently strong to develop a bias decomposition for $\phi_{1}$. However, since none of the units from the target population are observed, reasoning about $\phi_{2}$ is inherently more challenging. Thus, in order to develop a tractable bias decomposition for $\phi_{2}$ in Section \ref{sectiontransportability}, we posit two additional conditions. First, we introduce one mild assumption regarding the independence of the personal and neighborhood treatments.

\begin{assumption}[Conditional Treatment Independence]\label{assumptionindependenceofAandG}
We have
\begin{equation*}
A \ind G |S,\tilde{\mathbf{X}}, \tilde{\mathbf{U}}
\end{equation*}
where $\tilde{\mathbf{X}}$ and $\tilde{\mathbf{U}}$ are defined as in Assumption \ref{exchangeability}.
\end{assumption}
Assumption \ref{assumptionindependenceofAandG} only introduces modest limits on our framework, as the combination of Assumptions \ref{consistency} and \ref{exchangeability} already implicitly precludes most forms of dependence between $A$ and $\mathbf{A}_{\mathcal{N}}$. In Assumption \ref{consistency}, we assume that $Y^{(a)}$ equals $Y^{(a,G)}$, as opposed to $Y^{(a,G^{(a)})}$, a condition which would typically be violated by the existence of a direct causal pathway from $A$ to $\mathbf{A}_{\mathcal{N}}$. In Assumption \ref{exchangeability}, we assume that the common causes of $A$ and $\mathbf{A}_{\mathcal{N}}$ are sufficient to close the confounding pathway between $A$ and $Y$ through $\mathbf{A}_{\mathcal{N}}$; however, this would not hold true in most settings in which direct causal pathways from $\mathbf{A}_{\mathcal{N}}$ to $A$ exist. Finally, we introduce one assumption concerning the differences between the reference and target populations.

\begin{assumption}[Conditional Ignorability of Sampling]\label{assumptionoutcomegeneratingmechanism} For any $a \in \{0, 1\}$  we have
\begin{equation*}
   Y^{(a)} \ind S | G, \mathbf{\dot{X}},\mathbf{\dot{U}}
\end{equation*}
where $\mathbf{\dot{X}} = \mathbf{X}_{GY} \cup \mathbf{X}_{SY}$ and $\mathbf{\dot{U}} =  \mathbf{U}_{GY} \cup \mathbf{U}_{SY}$. 
\end{assumption}

Conditional ignorability of sampling is a commonly presented assumption in the transportability literature which precludes the possibility of trial engagement effects \citep{Huang2024, Ung2025}. We note that the presence of interference introduces the need to condition on $G$ and the common causes of $G$ and $Y$ (due to the collider structure), not just the common causes of $S$ and $Y$ for this assumption to hold. To alleviate the more stringent requirements introduced by the presence of interference, Assumption \ref{assumptionoutcomegeneratingmechanism} allows for the existence of a set of unobserved covariates which ensures ignorability, making this a relatively non-restrictive condition.

One graphical model consistent with our assumptions is depicted in Figure \ref{fig1}. Under Assumptions \ref{neighint}-\ref{assumptionoutcomegeneratingmechanism}, one could identify $\phi_{2}$ through a weighted estimator if they were to observe $G$ and $\mathbf{\Breve{U}}$.
\begin{proposition}[Identification Formula for $\phi_{2}$ ]\label{unbiasedIPW2}
Under Assumptions \ref{neighint} through \ref{assumptionoutcomegeneratingmechanism}, $\phi_{2}$ can be identified as
\begin{equation*}
    \phi_{2} = \mathbb{E} \Bigg[ \bigg( \frac{\mathbb{I}(A=1) \cdot Y}{p(A=1|S=1,\tilde{\mathbf{X}}, \tilde{\mathbf{U}})} -\frac{\mathbb{I}(A=0) \cdot Y}{p(A=0|S=1, \tilde{\mathbf{X}}, \tilde{\mathbf{U}})} \bigg) \cdot \frac{p(S=2|G, \mathbf{\breve{X}}, \mathbf{\breve{U}}) }{p(S=1|G, \mathbf{\breve{X}}, \mathbf{\breve{U}}) } \cdot \frac{p(S=1)}{p(S=2)} \bigg|S=1 \Bigg]
\end{equation*}
where $ \mathbf{\breve{X}} = \mathbf{\dot{X}} \cup \mathbf{X}_{AS} \cup \mathbf{X}_{AG}$ and $\mathbf{\breve{U}} = \mathbf{\dot{U}} \cup \mathbf{U}_{AS} \cup \mathbf{U}_{AG}$.
\end{proposition}

\begin{figure}[t]
\centering
        \begin{tikzpicture}
        \node[] (A) at (0,4) {$A_i$};
        \node[] (Y) at (3,4) {$Y_i$};
        \draw[-stealth,] (A) to (Y);
        
        \node[] (Uay) at (0.5, 5) {$\mathbf{U}_{i,AY}$};
        \draw[-stealth, bend left] (Uay) to (Y);
        \draw[-stealth, bend right] (Uay) to (A);

        \node[] (An) at   (0,2.5) {$A_{i,1}$};
        \node[] (An2) at (0,2) {$A_{i,2}$};
        \node[] (An3) at (0,1.5) {$A_{i,3}$};
        \node[] (An4) at (0,1) {$\vdots$};
        \node[] (An5) at (0.025,0.4) {$A_{i,n_i}$};
        \node[draw, minimum width=0.75cm, minimum height=2.6cm, anchor=south west] (Rectangle) at (-0.37,0.2) {};

        \node[] (G) at (1.5,2.75) {$G_i$};

        \draw[-stealth,] (An3) to (G);
        \draw[-stealth,] (G) to (Y);

        \node[] (Ugy) at (2.75,  2) {$\mathbf{U}_{i, GY}$};
        \draw[-stealth, bend right] (Ugy) to (Y);
        \draw[-stealth, bend left] (Ugy) to (Rectangle);

        \node[] (Uag) at (-1.5, 2.75) {$\mathbf{U}_{i,AG}$};
        \draw[-stealth, bend left] (Uag) to (A);
        \draw[-stealth, bend right] (Uag) to (Rectangle);

        \node[] (S) at (-3.5,2.75) {$S_i$};
        \draw[-stealth, bend left] (S) to (A);
        \draw[-stealth, bend right ] (S) to (Rectangle);

        \node[] (Uas) at (-2, 4.75) {$\mathbf{U}_{i,AS}$};
        \draw[-stealth, bend left] (Uas) to (A);
        \draw[-stealth, bend right] (Uas) to (S);

        \node[] (Ugs) at (-2.5016, 0.7068) {$\mathbf{U}_{i,GS}$};
        \draw[-stealth, bend left] (Ugs) to (S);
        \draw[-stealth, bend right] (Ugs) to (Rectangle);
        
        \node[] (Usy) at (-1.5, 6) {$\mathbf{U}_{i,SY}$};
        \draw[-stealth, bend right=35] (Usy) to (S);
        \draw[-stealth, bend left=35] (Usy) to (Y);

        \end{tikzpicture}
        \caption{ A directed acyclic graph consistent with our assumptions. $\mathbf{X}_i$ is omitted for legibility, $\{ A_{i,1},A_{i,2}, ..., A_{i,n_i}  \} = \mathbf{A}_{\mathcal{N}_i}$. We note that this graph is meant for general illustration and there are other systems which satisfy our assumptions (e.g., the six categories of the covariates do not have to be disjoint).}
        \label{fig1}
\end{figure}
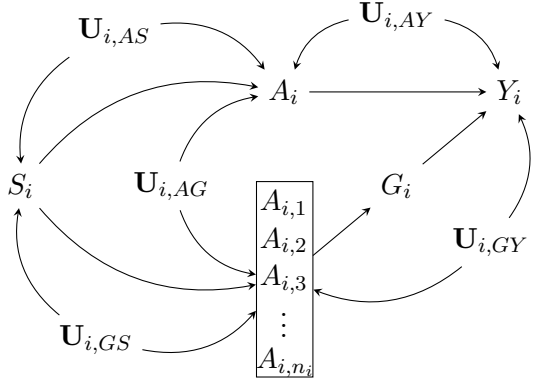

Propositions \ref{unbiasedIPW} and \ref{unbiasedIPW2} highlight that if all of the relevant variables were observed, both the average main treatment effect and the transported average main treatment effect could be identified through re-weighting formulations. However, the presence of spillover effects introduces the need to incorporate covariates inside the propensity score and sampling weight models that would typically be omitted without interference. This is particularly problematic, since many of these omitted variables might not just be individual-level covariates (e.g., a given school
district’s spending per student), but also neighborhood-level covariates (e.g., the average spending per
student of neighboring school districts), which are likely to not even be measured by a naive practitioner.

Propositions \ref{unbiasedIPW} and \ref{unbiasedIPW2} motivate our bias decompositions, as we will compare the expectation of a naive estimator (which ignores interference, unmeasured confounding, and lack of transportability) with the identifications formulas presented above. While we introduce $\phi_{1}$ and $\phi_{2}$ together in this section, to provide sufficient detail for both bias decompositions, we investigate them separately in Sections \ref{sensitivity} and \ref{sectiontransportability} respectively.

\section{Sensitivity Analysis for Unmeasured Confounding and Unobserved Interference}\label{sensitivity}
As discussed in Section \ref{secintroduction}, we examine two limitations to conducting causal inference in the presence of interference: failing to adjust for interference and failing to account for unobserved confounders. Both of these issues can introduce systematic bias into estimation; therefore, we examine both in tandem. To formalize, we consider the case where a practitioner tries to estimate the following functional, which, in the absence of interference and unmeasured confounding, equals $\phi_{1}$:
\begin{equation*} 
\psi = \mathbb{E} \Bigl[\frac{\mathbb{I}(A = 1) \cdot Y}{p(A = 1|S=1,  \mathbf{X}_{AY})}  -\frac{\mathbb{I}(A = 0) \cdot Y}{p(A = 0|S=1, \mathbf{X}_{AY})}  \Bigl{|}S=1\Bigl].
\end{equation*}
A typical estimator of $\psi$ is the IPW estimator \citep{Hirano2001},  given by 
\begin{equation*} 
\hat{\psi} = \mathbb{E}_n \Bigl[\frac{\mathbb{I}(A = 1, S=1) \cdot Y}{\hat{p}(A=1|S=1,  \mathbf{X}_{AY}) \cdot \hat{p}(S=1)} - \frac{\mathbb{I}(A = 0, S=1) \cdot Y}{\big(1 - \hat{p}(A=1|S=1,  \mathbf{X}_{AY}) \big) \cdot \hat{p}(S=1)}\Bigl], 
\end{equation*}
where $\mathbb{E}_n$ is the empirical expectation operator, $\hat{p}(S=1)$ is an estimate of $p(S=1)$, and $\hat{p}(A=1| \mathbf{X}_{AY}, S=1)$ is an estimator of $p(A=1| \mathbf{X}_{AY}, S=1)$. Throughout the rest of this section, we assume that the nuisance function estimator and $\hat{p}(S=1)$ are correctly specified and that the estimation error of $\hat{\psi}$ for $\psi$ is negligible. Therefore, we treat $\hat{\psi}$ as an unbiased estimator of $\psi$. Our goal is to study the bias that arises from using $\hat{\psi}$ instead of an unbiased estimator of $\phi_{1}$. Specifically, we are interested in the following quantity: 
\begin{align*}
\textup{\text{Bias}}_{\phi_{1}}(\hat{\psi}) &= \psi - \phi_{1}.
\end{align*}
As can be seen from the differences between the formula of $\psi$ and the IPW identification formula for $\phi_{1}$, the issue of ignored interference can be formulated as the more widely familiar problem of omitting certain causes of $A$ from the propensity score. Namely, in the presence of interference, in addition to failing to include unobserved confounders for the relationship between $A$ and $Y$ ($\mathbf{U}_{AY}$) inside the propensity score model, the naive practitioner also does not adjust for the \textit{common causes} of $A$ and $A_{\mathcal{N}}$ ($\mathbf{X}_{AG}$ and $\mathbf{U}_{AG}$) or the common causes of $A$ and $S$ ($\mathbf{X}_{AS}$ and $\mathbf{U}_{AS}$), resulting in a violation of Assumption \ref{exchangeability}. Thus, our sensitivity analysis enables practitioners to study the bias resulting from ignored interference solely by investigating the impacts of omitting the aforementioned covariates from the propensity score. 

\subsection{Bias Decomposition and Sensitivity Parameters}\label{biasdecomposition}
To conduct our sensitivity analysis, we begin by comparing $\psi$ with the unbiased IPW form of $\phi_{1}$ (as presented in Proposition \ref{unbiasedIPW}) through a multiplicative error in weights, given by
\begin{equation*}
\varepsilon_a = \frac{p(A=a|S=1,\tilde{\mathbf{X}}, \tilde{\mathbf{U}})}{p(A=a|S=1,\mathbf{X}_{AY})}.
\end{equation*}
We refer to $p(A=a|S=1, \tilde{\mathbf{X}}, \tilde{\mathbf{U}})$ as the \textit{true propensity score} for $A=a$ and to $p(A=a|S=1, \mathbf{X}_{AY})$ as the \textit{pseudo-propensity score} for $A=a$. We call $\varepsilon_a$ the multiplicative error in weights (MEW) score for $A=a$. $\varepsilon_a$ forms the basis of our sensitivity analysis, as $\textup{\text{Bias}}_{\phi_{1}}(\hat{\psi})$ depends on $\varepsilon_a$ and the potential outcomes. We formalize this in Theorem \ref{thmbiasdecomposition1}.
\begin{theorem}[Bias Decomposition for $\phi_{1}$]\label{thmbiasdecomposition1}
Under Assumptions \ref{neighint} through \ref{exchangeability}, the bias of the naive estimator, $\hat{\psi}$, for the natural average treatment effect, $\phi_{1}$, can be written as
\begin{align*}
\textup{\text{Bias}}_{\phi_{1}}(\hat{\psi}) &= \sum_{a = 0}^{1} (-1)^{1-a}\bigg( \rho_{Y^{(a,0)}, \varepsilon_a} \cdot \sigma_{Y^{(a,0)}} \cdot \sigma_{\varepsilon_a}\bigg) \tag{$T_{1}$} \\
&+ \sum_{a = 0}^{1} (-1)^{1-a} \bigg( \rho_{\gamma(a), \varepsilon_a} \cdot \sigma_{\gamma(a)} \cdot \sigma_{\varepsilon_a}  \bigg), \tag{$T_{2}$}
\end{align*}
where $\sigma_X$ represents the standard deviation of arbitrary variable $X$ conditional on $S=1$, and $\rho_{W, Z}$ represents the correlation between arbitrary variables $W$ and $Z$ conditional on $S=1$.
\end{theorem}
While the details of the proof are left for Appendix \ref{appendixproofs}, we provide some intuition behind Theorem \ref{thmbiasdecomposition1} below.

$T_1 \text{ and } T_2$ each highlight a different aspect of the bias that is introduced by ignored interference and unmeasured confounding. Namely, the bias of the naive estimator for $\phi_{1}$ is governed by a sum of covariances that reflect the relationship between MEW scores and two types of causal variables: the baseline potential outcomes and the natural spillover effects. We write each of these covariances as a product of a correlation term and two standard deviation terms, which form the basis of our sensitivity analysis framework. 

$T_1$ reflects the joint variability of MEW scores and the baseline potential outcomes, $Y^{(1,0)}$ and $Y^{(0,0)}$. We refer to $Y^{(1,0)}$ and $Y^{(0,0)}$ as the baseline potential outcomes since we view $G=0$ as the reference neighborhood treatment. Three terms contribute to $T_1$: the correlation between the baseline potential outcomes and MEW scores ($\rho_{Y^{(a,0)}, \varepsilon_a}$), the standard deviation of the baseline potential outcomes ($\sigma_{Y^{(a,0)}}$), and the standard deviation of MEW scores ($\sigma_{\varepsilon_a}$). If the variability of the baseline potential outcomes is high ($\sigma_{Y^{(a,0)}}$ is large), if the variance of the MEW scores is significant ($\sigma_{\varepsilon_a}$ is large), and if MEW scores and potential outcomes are highly correlated ($|\rho_{Y^{(a,0)}, \varepsilon_a}|$ is close to 1), $T_1$ would contribute significantly to the bias. 

$T_2$ is a measure of the joint variability of individual natural spillover effects ($\gamma(a)$) and MEW scores. $T_2$ contributes significantly to the bias if the following three conditions are met simultaneously: if there is a large amount of heterogeneity in the natural individual spillover effects ($\sigma_{\gamma(a)}^2$ is large), if there is significant variability in MEW scores ($\sigma_{\varepsilon_a}^{2}$ is large), and if there is a strong correlation between MEW scores and spillover effects ($|\rho_{\gamma(a), \varepsilon_a}|$ is close to 1).

\begin{remark} The presence of $T_2$ highlights the fact that under interference, even in the absence of common causes of $A$ and $A_{\mathcal{N}}$ and common causes of $A$ and $S$, the effects of unmeasured confounding on the bias still depend on the distribution of the individual natural spillover effects ($\gamma(a)$). Intuitively, $T_2$ adjusts for the fact that $Y^{(a)}$ is no longer equal to $Y^{(a,0)}$, but is now a function of $G$ and several distinct variables, $Y^{(a,g)}$ for $a \in \{0,1\}$ and $g \in \{0, 1, ..., g_{max}\}$. We note that in the absence of interference, our bias decomposition reduces to the one proposed by \cite{Shen2011}:
\begin{align*}
\textup{\text{Bias}}_{\phi_{1}}(\hat{\psi}) &= \sum_{a = 0}^{1} (-1)^{1-a}\bigg( \rho_{Y^{(a,0)}, \varepsilon^{'}_a} \cdot \sigma_{Y^{(a,0)}} \cdot \sigma_{ \varepsilon^{'}_a}\bigg).
\end{align*}
where $\varepsilon^{'}_a = \frac{p(A=a|\mathbf{X}_{AY}, \mathbf{U}_{AY}, S=1)}{p(A=a|S=1, \mathbf{X}_{AY})}$. 
\end{remark}

\subsection{Specifying the Sensitivity Parameters}\label{settingparams}
Next, we provide several guidelines for specifying the parameters of our bias decomposition. We note that these guidelines are meant to be widely applicable, and that in many cases, practitioners may have certain domain knowledge that allows them to specify these parameters in a less conservative manner.

\paragraph{Specifying Parameters in $T_1$} As discussed previously, the first term reflects the relationship between MEW scores and the baseline potential outcomes. Six sensitivity parameters need to be specified in $T_1$. 

$\sigma_{Y^{(1,0)}}$ and $\sigma_{Y^{(0,0)}}$  act as scaling factors in $T_1$. Unlike variables involving MEW scores,  the variability of the baseline potential outcomes depends only on the outcome variable of interest, and is therefore outside of the practitioner's control. Thus, we propose specifying a single upper bound for both values of $\sigma_{Y^{(a,0)}}$. For example, if the baseline potential outcomes are bounded between $y_{\text{min}}^{\text{ref}}$ and $y_{\text{max}}^{\text{ref}}$, by Popoviciu's inequality on variances, 
\begin{align*}
\max_{a \in \{0,1\}} \sigma_{Y^{(a,0)}} \leq \frac{1}{2} \cdot (y_{\text{max}}^{\text{ref}} - y_{\text{min}}^{\text{ref}}).
\end{align*}

 $\sigma_{\varepsilon_a}$ is a measure of variability in MEW scores. The mean of $\varepsilon_a$ is 1. If $\varepsilon_a$ deviates significantly from this value, that would imply that the pseudo-propensity score for $A=a$ significantly overestimates (if $\varepsilon_a < 1$) or underestimates (if $\varepsilon_a > 1$) the true propensity score (i.e., the effects of ignoring common causes and confounding are fairly strong). Following an argument analogous to the one presented by \cite{Shen2011}, an upper bound can be calculated for this quantity from the observed data, since for $a \in \{0,1\}$, $\sigma_{\varepsilon_a} \leq \sqrt{\mathbb{E} \big[ \frac{1-p(A=a|S=1, \mathbf{X}_{AY})}{p(A=a|S=1, \mathbf{X}_{AY})} | S=1 \big]}$. Unless additional information is available to the practitioner, we propose specifying an upper bound for both values of $\sigma_{\varepsilon_a}$ as some proportion of the maximum standard deviation of MEW scores through a single sensitivity parameter $\eta_{\varepsilon} \in [0,1]$:
 \begin{align*}
 \max_{a \in \{0,1\}}\sigma_{\varepsilon_a} \leq \eta_{\varepsilon} \cdot \sqrt{\mathbb{E} \bigg[ \frac{1-p(A=a|S=1, \mathbf{X}_{AY})}{p(A=a|S=1, \mathbf{X}_{AY})} \bigg| S=1 \bigg]}. 
\end{align*}
A formal benchmarking procedure could be employed to attain a reasonable range of values for $\eta_{\varepsilon}$.
\begin{remark}
    The choice of $\eta_{\varepsilon}$ can be viewed from the perspective of previous sensitivity analyses by \cite{Rosenbaum2002} and \cite{Tan2006}. Namely, instead of specifying $\eta_{\varepsilon}$, one could assume that $\varepsilon_a$ is bounded between $\alpha_a^{-1}$ and $\alpha_a$, where $\alpha_a \geq 1$ is a sensitivity parameter. In that case, by the Bhatia-Davis inequality, $\sigma_{\varepsilon_a} \leq \sqrt{\alpha_a - 2 + \alpha_a^{-1}}$. A choice of $\alpha_a$ which provides a non-trivial bound on $\sigma_{\varepsilon_a}$ in light of the distribution of the pseudo-propensity scores can be viewed as corresponding to $\eta_{\varepsilon} < 1$. A value of 1 for $\alpha_a$ and 0 for $\eta_{\varepsilon}$ corresponds to the setting in which the pseudo-propensity score is identical to the true propensity score (i.e., no confounding or interference).
\end{remark}
The two $\rho_{Y^{(a,0)}, \varepsilon_a}$ values describe the correlation between the baseline potential outcomes and MEW scores, and determine the sign of $T_1$. If these correlations are zero, then $T_1$ does not contribute to the bias, despite the presence of unmeasured confounding and interference. Generally, a positive $\rho_{Y^{(a,0)}, \varepsilon_a}$ implies that units with higher values of $Y^{(a,0)}$ tend to have pseudo-propensity scores that are lower than their true propensity scores; meanwhile, a negative  $\rho_{Y^{(a,0)}, \varepsilon_a}$ indicates that the pseudo-propensity scores of units with high values of $Y^{(a,0)}$ tend to overestimate their true propensity scores. If a practitioner does not have domain knowledge which allows them to specify the signs of the correlations, they can first set these signs so as to maximize bias, and then vary their magnitude from 0 to 1 by setting $\rho_{\text{baseline}} = |\rho_{Y^{(0,0)}, \varepsilon_0}| = |\rho_{Y^{(1,0)}, \varepsilon_1}|$.

\paragraph{Specifying Parameters in $T_2$} $T_2$ describes the relationship between the individual natural spillover effects and MEW scores. To determine the magnitude and sign of $T_2$, four additional terms (beyond those present in $T_1$) need to be specified.

$\sigma_{\gamma(a)}$ measures the variability of natural spillover effects if $G$ is allowed to vary as it does in the reference population. This term is large if there is a significant number of units which receive neighborhood treatments other than 0 and if there is notable variability in causal contrasts of the form $Y^{(a,g)} - Y^{(a,0)}$. If the natural spillover effects ($\gamma(a)$) are both bounded between $x_{\text{min}}^{\text{ref}}$ and $x_{\text{max}}^{\text{ref}}$ for $a \in \{0,1\}$, the worst-case bound for this term is $ \frac{1}{2} \cdot (x_{\text{max}}^{\text{ref}} - x_{\text{min}}^{\text{ref}})$. While this is another quantity the practitioner has no control over, in many settings where interference is ignored, researchers may be interested in estimating the bias of their naive estimator under a variety of assumptions about the strength of the ignored interference. Therefore, we suggest specifying an upper bound for both values of $\sigma_{\gamma(a)}$ as a single proportion of the maximum standard deviation of $\gamma(a)$ through the sensitivity parameter $\eta_{\gamma} \in [0,1]$, which reflects the level of heterogeneity in the individual spillover effects:
\begin{align*}
\max_{a \in \{0,1\}} \sigma_{\gamma(a)} \leq \frac{\eta_{\gamma}}{2} \cdot (x_{\text{max}}^{\text{ref}} - x_{\text{min}}^{\text{ref}}).  
\end{align*}

 $\rho_{\gamma(a), \varepsilon_a}$ determines the sign of $T_2$. Generally, if the correlation is positive, then the pseudo-propensity scores of units with higher values of $\gamma(a)$ usually underestimate their true propensity scores, and if the correlation is negative, units with higher values of $\gamma(a)$ have pseudo-propensity scores that typically overestimate their true propensity scores. If domain knowledge does not allow a practitioner to specify these correlations, we propose setting the signs of $\rho_{\gamma(a), \varepsilon_a}$ so as to maximize bias, and adjusting their common magnitude $\rho_{\text{spillover}} = |\rho_{\gamma(0), \varepsilon_0}| = |\rho_{\gamma(1), \varepsilon_1}|$ from 0 to 1.

\section{Sensitivity Analysis for Transportability with Interference}\label{sectiontransportability}
In this section, we consider settings in which the practitioner is interested in estimating the transported natural average treatment effect ($\phi_{2}$) alongside the natural average treatment effect ($\phi_{1}$). Because treatment assignment mechanisms and covariate distributions may differ between the target and reference populations, in general, $\phi_{1} \neq \phi_{2}$. If the practitioner wrongly assumes that their naive estimate generalizes to the target population, they will incur bias due to lack of transportability, in addition to any existing bias arising from ignoring interference and unmeasured confounding \citep{Degtiar2023}. Proposition \ref{unbiasedIPW2} highlights that bias due to lack of transportability occurs since the practitioner's naive estimator, $\hat{\psi}$, only incorporates information from the reference population ($S=1$), without reweighting it to the target population ($S=2$). 

To estimate the bias arising from lack of transportability, we introduce a new term: a difference error in weights (DEW) score, denoted by $\tilde{\varepsilon}_S$, and given by
\begin{equation*}
    \tilde{\varepsilon}_S = 1 - \frac{p(S=2|G, \mathbf{\breve{X}}, \mathbf{\breve{U}}) }{p(S=1|G, \mathbf{\breve{X}}, \mathbf{\breve{U}}) } \cdot \frac{p(S=1)}{p(S=2)}.
\end{equation*}
While the MEW score, $\varepsilon_a$, compares the naive propensity score with the true propensity score on the multiplicative scale, $\tilde{\varepsilon}_S$ compares the naive sampling weight of 1 to the true sampling weight on the difference scale. We note that if there is no transportability bias, $\tilde{\varepsilon}_S = 0$. To write the bias of a naive estimator for $\phi_2$ as a function of DEW scores, we decompose it into the sum of the bias of $\hat{\psi}$ for $\phi_1$ and the difference between $\phi_1$ and $\phi_2$:
\begin{align*}
\textup{\text{Bias}}_{\phi_{2}}(\hat{\psi}) &= \psi - \phi_{2} = \big(\psi - \phi_{1} \big) + \big(\phi_{1} -\phi_{2} \big).
\end{align*}
The second equality forms the basis of our bias decomposition, which is formalized below.

\begin{theorem}[Bias Decomposition for $\phi_{2}$]\label{thmbiasdecomposition2}
Under Assumptions \ref{neighint} through \ref{assumptionoutcomegeneratingmechanism}, the bias of the naive estimator, $\hat{\psi}$, for the transported natural average treatment effect, $\phi_{2}$, can be written as
\begin{align*}
\textup{\text{Bias}}_{\phi_{2}}(\hat{\psi}) &= \textup{\text{Bias}}_{\phi_{1}}(\hat{\psi}) \tag{$T_{1} + T_{2}$} \\
&+ \rho_{\tau(0), \tilde{\varepsilon}_S} \cdot \sigma_{\tau(0)} \cdot \sigma_{\tilde{\varepsilon}_S}  \tag{$T_{3}$} \\
&+ \rho_{\gamma(1)-\gamma(0), \tilde{\varepsilon}_S} \cdot \sigma_{\gamma(1)-\gamma(0)} \cdot \sigma_{\tilde{\varepsilon}_S}  \tag{$T_{4}$}
\end{align*}
where  $T_1$, $T_2$, $\sigma_X$, and $\rho_{W,Z}$ are defined as in Theorem \ref{thmbiasdecomposition1}.
\end{theorem}
Theorem \ref{thmbiasdecomposition2} underlines that the difference between $\phi_{1}$ and $\phi_{2}$ is governed by two terms: $T_3$, which captures the relationship between DEW scores and direct effects (namely, controlled individual direct effects at $G=0$), and $T_4$, which captures the relationship between DEW scores and indirect effects (in particular, the difference in the natural individual spillover effects). 

$T_3$ represents the joint variability of $\tau(0)$ and $\tilde{\varepsilon}_S$, and is a product of three sensitivity parameters. The first is the correlation between the controlled individual main effects at $G=0$ and DEW scores ($\rho_{\tau(0), \tilde{\varepsilon}_s}$). If the correlation is positive, units with higher values of $\tau(0)$ tend to have sampling weights less than 1; if the correlation is negative, units with higher values of $\tau(0)$ generally have sampling weights greater than 1. The second sensitivity parameter is $\sigma_{\tau(0)}$, which is a scaling factor that reflects the impact of personal treatment effect heterogeneity on the bias. The third parameter is $\sigma_{\tilde{\varepsilon}_S}$, which measures the variability in DEW scores around their mean value of 0, acting as a gauge for the impact of the sampling weights.

$T_4$ is the covariance between DEW scores and the difference in the natural spillover effects at $A=1$ and $A=0$. The magnitude of $T_4$ depends on the heterogeneity in the difference of the individual natural spillover effects ($\sigma_{\gamma(1)-\gamma(0)}$), the variability of DEW scores ($\sigma_{\tilde{\varepsilon}_S}$), and the correlation between $\gamma(1)-\gamma(0)$ and DEW scores. This last term also determines the sign of $T_4$: if $\rho_{\gamma(1)-\gamma(0), \tilde{\varepsilon}_S} > 0$, units with larger spillover effects at $A=1$ when compared with $A=0$ typically have sampling weights less than 1, while if $\rho_{\gamma(1)-\gamma(0), \tilde{\varepsilon}_S} < 0$, they tend to have sampling weights greater than 1. As is the case with $T_3$, all three of the sensitivity parameters comprising $T_4$ have to be significant for the term to contribute meaningfully to the bias.

\begin{remark}
    We note that $\textup{\text{Bias}}_{\phi_{2}}(\hat{\psi})$ can also be written as 
    \begin{align*}
\textup{\text{Bias}}_{\phi_{2}}(\hat{\psi}) &= \sum_{a = 0}^{1} (-1)^{1-a}\bigg( \rho_{Y^{(a,0)}, \varepsilon_a + \tilde{\varepsilon}_S} \cdot \sigma_{Y^{(a,0)}} \cdot \sigma_{\varepsilon_a + \tilde{\varepsilon}_S}\bigg) \\
&+ \sum_{a = 0}^{1} (-1)^{1-a} \bigg( \rho_{\gamma(a), \varepsilon_a + \tilde{\varepsilon}_S} \cdot \sigma_{\gamma(a)} \cdot \sigma_{\varepsilon_a + \tilde{\varepsilon}_S}  \bigg).
\end{align*}
Thus, the bias of $\hat{\psi}$ for $\phi_2$ can be represented as the joint variability of the combined error term (the sum of MEW and DEW scores) with the baseline potential outcomes and individual natural spillover effects.
\end{remark}

The bias decomposition outlined above is widely applicable. Nevertheless, Theorem \ref{thmbiasdecomposition2} can also be significantly simplified in many settings. We outline several such cases below.

\begin{corollary}
    If $S$ and $Y$ as well as $G$ and $Y$ have no common causes, $\tilde{\varepsilon}_S $ will reduce to
    \begin{align*}
\tilde{\varepsilon}_S = 1-\frac{p(S=2|G)}{p(S=1|G)} \cdot \frac{p(S=1)}{p(S=2)} = 1 - \frac{p(G|S=2)}{p(G|S=1)}.
\end{align*}
Therefore, the transportability bias would be equal to 0 if the marginal distribution of $G$ was the same between the two populations. 
\end{corollary}

\begin{corollary} \label{corollarynointerference}
    In the absence of interference, the bias caused by lack of transportability reduces to $T_3$, and $\tilde{\varepsilon}_S $ becomes 
            \begin{align*}
            \tilde{\varepsilon}_S = 1 - \frac{p(S=2|\mathbf{X}_{SY}, \mathbf{U}_{SY}) }{p(S=1|\mathbf{X}_{SY}, \mathbf{U}_{SY}) } \cdot \frac{p(S=1)}{p(S=2)}.
    \end{align*}
\end{corollary}
Corollary \ref{corollarynointerference} highlights that our method generalizes the sensitivity analysis framework presented by \cite{Huang2024}.

\begin{proposition}\label{propositionoutcomegeneratingequation}
    If Assumptions \ref{neighint} through \ref{assumptionoutcomegeneratingmechanism} are met, $Y$ and $S$ have no common causes, $S$ only has direct causal effects on $A$ and $G$, and $Y$ is generated according to the model
    \begin{equation*}
        Y = f_0(\mathbf{X}, \mathbf{U}) + f_1(A, \mathbf{X}, \mathbf{U}) + f_2(G, \mathbf{X}, \mathbf{U}) + \epsilon(A,G),
    \end{equation*}
    where $\epsilon(A,G)$ is a zero-mean noise term uncorrelated with any of the covariates, then $\phi_1 = \phi_2$.
\end{proposition}

In order to specify the values of the sensitivity parameters in $T_3$ and $T_4$, we extend our methodology from Section \ref{sensitivity}. Since the personal treatment effect heterogeneity depends only on the outcome variable, assuming that the potential outcomes are bounded between $y_{\text{min}}^{\text{ref}}$ and $y_{\text{max}}^{\text{ref}}$, we propose setting an upper bound for $\sigma_{\tau(0)}$ as
\begin{equation*}
    \sigma_{\tau(0)} \leq (y_{\text{max}}^{\text{ref}} - y_{\text{min}}^{\text{ref}}).
\end{equation*}
On the other hand, since practitioners may wish to evaluate the bias under a variety of interference patterns, the upper bound on the standard deviation of $\gamma(1)-\gamma(0)$ can be specified as a proportion of its maximum value via the parameter $\beta_{\gamma} \in [0,1]$, which reflects the practitioners' beliefs about the variability of the difference in spillover effects:
\begin{equation*}
    \sigma_{\gamma(1)-\gamma(0)} \leq \beta_{\gamma} \cdot  (x_{\text{max}}^{\text{ref}} - x_{\text{min}}^{\text{ref}}).
\end{equation*}

Similarly, the practitioner can specify a value $\alpha_S \geq 1$ such that the sampling weights lie between $\alpha_S^{-1}$ and $\alpha_S$, which results in an upper bound for the value of $\sigma_{\tilde{\varepsilon}_S}$ via the Bhatia-Davis inequality:
\begin{equation*}
    \sigma_{\tilde{\varepsilon}_S} \leq \frac{(\alpha_S-1)}{\sqrt{\alpha_S}}.
\end{equation*}

As before, in the absence of additional information, the practitioner can set the signs of the correlations ($\rho_{\tau(0), \tilde{\varepsilon}_S}$ and $\rho_{\gamma(1)-\gamma(0), \tilde{\varepsilon}_S}$) so as to maximize bias, and vary their magnitude between 0 and 1. The bias resulting from various combinations of the aforementioned sensitivity parameters can be examined via two and three-dimensional contour plots.

\section{Bias Decomposition in the Presence of Undefined Potential Outcomes}\label{undefined}
In earlier sections of this paper, we considered settings in which all potential outcomes are well-defined for all units. While Assumption \ref{assumptionwelldefined}, which encodes this condition, is widely applicable, there may be certain settings in which it is unreasonable. Concerns about violations of Assumption \ref{assumptionwelldefined} are most relevant for applications where the number of neighbors (which determines the domain of the exposure mapping function) is viewed as fixed. For instance, in a geographic context, if $G$ represents the number of a neighboring countries adopting a certain policy, defining the potential outcome $Y^{(1,2)}$ for a country that has only one neighbor could be seen as conceptually incoherent. Moreover, in general, potential outcomes of the form $Y^{(a,g)}$ are undefined for units without neighbors \citep{Forastiere2021, Kim2025}.  Due to these issues, practitioners may wish to define a different number of potential outcomes for each unit to better match their philosophical intuitions about causality. 

In this section, we extend our earlier results to account for this complexity. Without loss of generality, we will assume that $G \in \{0, 1, 2, ..., g_{max}\}$ represents the number of treated neighbors a unit has. For the sake of notational convenience, we also make the following assumptions:
\begin{itemize}
    \item $Y^{(a,g)}$ for $a \in \{0,1\}$ is well-defined for any unit with $g$ or more neighbors, and is undefined for any unit with fewer than $g$ neighbors.
    \item  $N$ records the number of neighbors a unit has and $N \in \{0, 1, 2, ..., g_{max} \}$.
    \item The potential outcomes $Y^{(1,0)}$ and $Y^{(0,0)}$ are well-defined for units with $N = 0$ (i.e., the outcomes of units without neighbors behave like those of units that have neighbors none of whom are treated). 
    \item   $V_g = \mathbb{I}(N \geq g)$ tracks units for which $Y^{(a,g)}$ is well-defined. Note that for any $\tilde{g} < g'$ if  $V_{g'} = 1$, then $V_{\tilde{g}} = 1$.  
\end{itemize}

In order to extend our bias decomposition to this setting, the assumptions and causal estimands outlined in Section \ref{problem} need to be modified. First, we introduce a new version of consistency, which relates the observed outcome to the potential outcomes in a manner which depends on a unit's number of neighbors.
\begin{assumption}[Consistency with Undefined Potential Outcomes] \label{assumptionconditionalconsistency}
Given that $N = n$, the potential outcome $Y^{(a)}$ satisfies $Y^{(a)} = \sum_{g = 0}^{n} \mathbb{I}(G=g) \cdot  Y^{(a,g)}$ and the observed outcome $Y$ satisfies $Y =  \sum_{a = 0}^{1} \mathbb{I}(A=a) \cdot  Y^{(a)}$.
\end{assumption}
Since potential outcomes are only well-defined conditionally under this framework, the causal estimands need to be re-defined. We begin by considering a measure of the causal impact of changing $A$, the \textit{controlled individual main effect}, 
\begin{equation*}
\tau(g) :=  \begin{cases} 
      \text{N.A.   if } V_g = 0 \\
      Y^{(1,g)}-Y^{(0,g)} \text{   if } V_g = 1. \\ 
   \end{cases}
\end{equation*}
Here, $\tau(g)$ still represents the impact of changing $A$ from 0 to 1 while $G$ is fixed. However, this quantity is only well-defined for units who have enough neighbors to receive neighborhood treatment $g$. The other causal estimand measuring the impact of changing $A$, the \textit{natural individual main effect}, $\kappa$, now also depends on the number of neighbors a unit has. Namely, if $N=n$, then $\kappa = Y^{(1)} - Y^{(0)} = \sum_{g = 0}^{n} \mathbb{I}(G=g) \cdot  \tau(g)$. 

Our previous parameters of interest, the \textit{natural average main effect}, $\phi_{1}$, and the \textit{transported natural average main effect}, $\phi_{2}$, can be expressed as conditional expectations of $\kappa$ across the reference and target populations respectively, as
\begin{align*}
\phi_{s} &= \sum_{n=0}^{ g_{\text{max}}}\mathbb{E}\big[\underbrace{Y^{(1)}-Y^{(0)}}_{\kappa} |S=s, N =n \big] \cdot p(N=n|S=s) \\
&= \sum_{n=0}^{ g_{\text{max}}} \sum_{g=0}^{ n}\mathbb{E}\big[\mathbb{I}(G=g) \cdot \tau(g) |S=s, N =n \big] \cdot p(N=n|S=s)\\
&= \sum_{g=0}^{ g_{\text{max}}} \sum_{n=g}^{ g_{\text{max}}} \mathbb{E}\big[\mathbb{I}(G=g) \cdot \tau(g) |S=s, N =n\big] \cdot p(N=n|S=s) \\
&= \sum_{g = 0}^{ g_{\text{max}}} \mathbb{E}\big[\mathbb{I}(G=g) \cdot \tau(g) |S=s,  V_g = 1 \big] \cdot p(V_g=1|S=s).
\end{align*}

Next, we modify our assumptions of positivity, exchangeability, and conditional treatment independence to account for undefined potential outcomes.
\begin{assumption}[Positivity with Undefined Potential Outcomes]\label{assumptionconditionalpositivity}
For any $ a \in \{0, 1\}$ and any $ s \in \{1,2\}$ we have, almost surely in $(N, \mathbf{X}, \mathbf{U})$ and  $(N, G, \mathbf{X}, \mathbf{U})$,
\begin{equation*}
p(A=a|S = s, N, \mathbf{X}, \mathbf{U}) > 0  \text{ and } p(S=s|N, G,\mathbf{X}, \mathbf{U}) > 0.
\end{equation*}

\end{assumption}
\begin{assumption}[Weak Conditional Exchangeability with Undefined Potential Outcomes]\label{assumptionconditionalexchangeability}
Given that $N=n$ where $n \in \{0,1,2,..., g_{\text{max}}\}$, for any $a \in \{0, 1\}$, we have 
\begin{equation*}
Y^{(a)} \ind A|S, N=n, \tilde{\mathbf{X}}, \tilde{\mathbf{U}}.
\end{equation*}
\end{assumption}

Though Assumptions \ref{assumptionconditionalconsistency} through \ref{assumptionconditionalexchangeability} suffice to develop a bias decomposition for $\phi_{1}$, we once again require the conditions of treatment independence and sampling ignorability to develop a bias decomposition for $\phi_{2}$.

\begin{assumption}[Conditional Treatment Independence with Undefined Potential Outcomes]\label{assumptionindependenceofAandGundefinedpotentialoutcomes}
We have
\begin{equation*}
A \ind G |S,N,\tilde{\mathbf{X}}, \tilde{\mathbf{U}}.
\end{equation*}
\end{assumption}

\begin{assumption}[Conditional Ignorability of Sampling with Undefined Potential Outcomes]\label{assumptionsamplingignorabilityundefined} Given that $N=n$, for any $a \in \{0, 1\}$  we have
\begin{equation*}
   Y^{(a)} \ind S |N=n, G, \mathbf{\dot{X}},\mathbf{\dot{U}}.
\end{equation*}
\end{assumption}

Finally, we introduce one additional condition, which allows us to deal with the difficulties introduced by the presence of undefined potential outcomes, and to recover bias decomposition formulas similar to the ones in Sections  \ref{sensitivity} and \ref{sectiontransportability}.

\begin{assumption}[Independence Assumptions for Undefined Potential Outcomes]\label{assumptionindependenceforundefinedoutcomes} We assume the following:
\begin{itemize}
    \item The two nuisance functions for $A$, the true propensity score and the pseudo-propensity score, do not depend on the number of neighbors. In other words, $p(A=a|S=1, N=n, \mathbf{\tilde{X}}, \mathbf{\tilde{U}}) = p(A=a|S=1, \mathbf{\tilde{X}}, \mathbf{\tilde{U}})$ and $p(A=a|S=1, N=n,  \mathbf{X}_{AY}) = p(A=a|S=1,  \mathbf{X}_{AY})$.
    \item The nuisance function for $S$ does not depend on the number of neighbors. In other words, $p(S=s|N=n, G, \mathbf{\breve{X}},  \mathbf{\breve{U}}) = p(S=s|G, \mathbf{\breve{X}},  \mathbf{\breve{U}})$.
    \item The marginal distribution of the number of neighbors is equivalent between the two populations (i.e., $ N \ind S$).
\end{itemize}

\end{assumption}

Using the assumptions outlined above, we can extend our previous formulas providing identification for $\phi_1$ and $\phi_2$ (if all of the relevant variables were observed) to our new setting.

\begin{proposition}[Identification Formulas with Undefined Potential Outcomes]\label{unbiasedIPW1} Under Assumption \ref{neighint}, Assumptions \ref{assumptionconditionalconsistency} through \ref{assumptionconditionalexchangeability}, and Assumption \ref{assumptionindependenceforundefinedoutcomes}, $\phi_{1}$ can be identified by the formula given in Proposition \ref{unbiasedIPW}. Meanwhile, under Assumption \ref{neighint} and Assumptions \ref{assumptionconditionalconsistency} through \ref{assumptionindependenceforundefinedoutcomes}, $\phi_{2}$ can be identified by the formula given in Proposition \ref{unbiasedIPW2}.
\end{proposition}

We note that although our identification formulas do not change, the fact that $Y^{(a)}$ is defined differently depending on a unit's number of neighbors necessitates stricter assumptions, and introduces unique challenges (such as requiring algebraic manipulations of the potential outcomes to be carried out conditionally on $N = n$). As in Sections \ref{sensitivity} and \ref{sectiontransportability}, the practitioner's naive estimate can be compared to the identification formulas provided in Proposition \ref{unbiasedIPW1} through several covariance terms relating MEW and DEW scores with the potential outcomes. Using this insight, we develop a bias decomposition appropriate for settings with undefined potential outcomes below. 

\begin{theorem}[Bias Decomposition in Presence of Undefined Potential Outcomes]\label{thmupdatedbiasdecomposition}
Under Assumption \ref{neighint}, Assumptions \ref{assumptionconditionalconsistency} through \ref{assumptionconditionalexchangeability}, and Assumption \ref{assumptionindependenceforundefinedoutcomes}, the bias of a naive estimator, $\hat{\psi}$, for the natural average treatment effect, $\phi_{1}$, can be written as
\begin{align*}
\text{\textup{Bias}}_{\phi_{1}}(\hat{\psi}) &= \sum_{a = 0}^{1}  (-1)^{1-a} \cdot \text{\textup{Cov}} \Bigl( Y^{(a,0)}, \varepsilon_a |S=1\Bigl) \\
&+ \sum_{a = 0}^{1}\sum_{n = 0}^{ g_{\text{max}}} (-1)^{1-a} \cdot \text{\textup{Cov}} \Bigl(\gamma(a) , \varepsilon_a |S=1, N=n \Bigl) \cdot p(N=n|S=1),
\end{align*}
where $\varepsilon_a$ is defined as in Theorem \ref{thmbiasdecomposition1}. Meanwhile, under Assumption \ref{neighint} and Assumptions \ref{assumptionconditionalconsistency} through \ref{assumptionindependenceforundefinedoutcomes}, the bias of $\hat{\psi}$ for the transported natural main effect, $\phi_{2}$, can be written as
\begin{align*}
\textup{\text{Bias}}_{\phi_{2}}(\hat{\psi}) &= \textup{\text{Bias}}_{\phi_{1}}(\hat{\psi}) \\
&+\textup{Cov}\big(\tau(0), \tilde{\varepsilon}_S \big| S=1\big)  \\
&+ \sum_{n = 0}^{ g_{\text{max}}} \textup{Cov}\big(\gamma(1)-\gamma(0), \tilde{\varepsilon}_S\big| S=1, N=n \big) \cdot p(N=n|S=1),
\end{align*}
where $\tilde{\varepsilon}_S$ is defined as in Theorem \ref{thmbiasdecomposition2}.
\end{theorem}
Theorem \ref{thmupdatedbiasdecomposition} can serve as the basis of a sensitivity analysis framework developed according to the principles discussed in Section \ref{sensitivity} and Section \ref{sectiontransportability}. Note that if the potential outcomes $Y^{(a)}$ for $a \in \{0,1\}$ are well-defined for all units (without conditioning on $N=n$), the bias decompositions in Theorem \ref{thmupdatedbiasdecomposition} reduce to the forms provided in Theorems \ref{thmbiasdecomposition1}  and \ref{thmbiasdecomposition2}. Thus, many of the results from  previous sections of this paper can be extended to applications with undefined potential outcomes.

\section{Conclusion}\label{sectionconclusion}
We introduced a sensitivity analysis framework, which, unlike existing frameworks, can be used to explore the effects of unmeasured confounding, omitted interference, and lack of transportability simultaneously. Our method does not require strict parametric assumptions about the data generating mechanism and provides practitioners with the ability to specify bias through several easily interpretable sensitivity parameters. These parameters are flexible and can integrate a wide array of perspectives informed by domain knowledge. We also investigated several special cases under which our bias decomposition can be significantly simplified. Finally, we extended our methodology to account for the additional data complexity of undefined potential outcomes, which can arise in the presence of interference.

\bibliography{bibliography}
\bibliographystyle{apalike}

\appendix
\appendixpage

\section{Proofs}\label{appendixproofs}
\begin{proof}[Proof of Proposition \ref{unbiasedIPW}] We begin by considering $\mathbb{E} \big[Y^{(a)}|S=1 \big]$. We note that
\begin{align*}
    \mathbb{E} \big[Y^{(a)}|S=1 \big] &= \mathbb{E} \Bigl[ \mathbb{E} \big[Y^{(a)}\big|S=1, \tilde{\mathbf{X}}, \tilde{\mathbf{U}}\big] \Bigl|S=1 \Bigl] \\
    &= \mathbb{E} \Bigl[ \mathbb{E} \big[Y^{(a)}\big|S=1, A=a, \tilde{\mathbf{X}}, \tilde{\mathbf{U}}\big]\Bigl|S=1 \Bigl] \\
    &= \mathbb{E} \Bigl[ \mathbb{E} \big[Y\big|S=1, A=a, \tilde{\mathbf{X}}, \tilde{\mathbf{U}}\big]\Bigl|S=1 \Bigl] \\
    &= \int_{\tilde{\mathbf{x}}, \tilde{\mathbf{u}}, y} y \cdot p(y|S=1, A=a,\tilde{\mathbf{x}}, \tilde{\mathbf{u}}) \cdot p(\tilde{\mathbf{x}}, \tilde{\mathbf{u}} | S=1) \\
    &= \int_{\tilde{\mathbf{x}}, \tilde{\mathbf{u}}, y} y \cdot \frac{p(y, A=a, \tilde{\mathbf{x}}, \tilde{\mathbf{u}}|S=1)}{p(A=a|S=1, \tilde{\mathbf{x}}, \tilde{\mathbf{u}})} \\
    &= \int_{\tilde{\mathbf{x}}, \tilde{\mathbf{u}}, y, \tilde{a}} \mathbb{I}(\tilde{a}=a) \cdot y \cdot \frac{p(y, \tilde{a}, \tilde{\mathbf{x}}, \tilde{\mathbf{u}}|S=1)}{p(A=a|S=1, \tilde{\mathbf{x}}, \tilde{\mathbf{u}})} \\
    &= \mathbb{E} \bigg[ \frac{\mathbb{I}(A=a) \cdot Y}{p(A=a|S=1, \tilde{\mathbf{X}}, \tilde{\mathbf{U}})} \bigg| S=1 \bigg].
\end{align*}
Utilizing the argument above for $\mathbb{E} \big[Y^{(1)}|S=1 \big]$ and $\mathbb{E} \big[Y^{(0)}|S=1 \big]$ proves Proposition \ref{unbiasedIPW}.
\end{proof}

\begin{proof}[Proof of Proposition \ref{unbiasedIPW2}] First, we let $\mathfrak{X} = \mathbf{\tilde{X}} \cup \mathbf{\dot{X}} $ and $\mathfrak{U} = \mathbf{\tilde{U}} \cup \mathbf{\dot{U}}$, and note that by the definitions of the covariate sets, we have $Y^{(a)}\ind S|G, \mathfrak{X}, \mathfrak{U}$ and $Y^{(a)} \ind A | S,G, \mathfrak{X}, \mathfrak{U}$. Thus, we have
\begin{align*}
    \mathbb{E} \big[Y^{(a)}|S=2 \big] &= \mathbb{E} \Bigl[ \mathbb{E} \big[Y^{(a)}\big|S=2,G, \mathfrak{X}, \mathfrak{U}\big] \Bigl|S=2 \Bigl] \\
    &= \mathbb{E} \Bigl[ \mathbb{E} \big[Y^{(a)}\big|S=1,G, \mathfrak{X}, \mathfrak{U}\big] \Bigl|S=2 \Bigl] \\
    &= \mathbb{E} \Bigl[ \mathbb{E} \big[Y^{(a)}\big|S=1,A=a,G, \mathfrak{X}, \mathfrak{U}\big] \Bigl|S=2 \Bigl] \\
    &= \int_{g,\mathfrak{x},\mathfrak{u},y} y \cdot p(y|S=1, A=a, g, \mathfrak{x}, \mathfrak{u}) \cdot p(g,\mathfrak{x},\mathfrak{u}|S=2) \\
    &= \int_{g,\mathfrak{x},\mathfrak{u},y} y \cdot p(y|S=1, A=a, g, \mathfrak{x}, \mathfrak{u}) \cdot \frac{p(g, \mathfrak{x},\mathfrak{u}|S=2)}{p(g, \mathfrak{x},\mathfrak{u}|S=1)} \cdot p(g, \mathfrak{x},\mathfrak{u}|S=1) \\
    &= \int_{g,\mathfrak{x},\mathfrak{u},y} y \cdot \frac{p(y, A=a, g, \mathfrak{x}, \mathfrak{u}|S=1)}{p(A=a|S=1, g, \mathfrak{x},\mathfrak{u})} \cdot \frac{p(g, \mathfrak{x},\mathfrak{u}|S=2)}{p(g, \mathfrak{x},\mathfrak{u}|S=1)} \\ 
    &= \int_{g,\mathfrak{x},\mathfrak{u},y} y \cdot \frac{p(y, A=a, g, \mathfrak{x}, \mathfrak{u}|S=1)}{p(A=a|S=1, g, \mathfrak{x},\mathfrak{u})} \cdot \frac{p(S=2|g, \mathfrak{x},\mathfrak{u})}{p(S=1|g, \mathfrak{x},\mathfrak{u})} \cdot \frac{p(S=1)}{p(S=2)} \\
    &= \int_{g,\mathfrak{x},\mathfrak{u},y, \tilde{a}} \mathbb{I}(\tilde{a}=a) \cdot y \cdot \frac{p(y, \tilde{a}, g, \mathfrak{x}, \mathfrak{u}|S=1)}{p(A=a|S=1, g, \mathfrak{x},\mathfrak{u})} \cdot \frac{p(S=2|g, \mathfrak{x},\mathfrak{u})}{p(S=1|g, \mathfrak{x},\mathfrak{u})} \cdot \frac{p(S=1)}{p(S=2)}.
\end{align*}
Therefore, we have 
\begin{align*}
    \mathbb{E} \big[Y^{(a)}|S=2 \big] &= \mathbb{E} \Bigg[ \frac{\mathbb{I}(A=a) \cdot Y}{p(A=a|S=1,G,\mathfrak{X}, \mathfrak{U})}  \cdot \frac{p(S=2|G, \mathfrak{X}, \mathfrak{U}) }{p(S=1|G, \mathfrak{X}, \mathfrak{U}) } \cdot \frac{p(S=1)}{p(S=2)} \bigg|S=1 \Bigg] \\
    &= \mathbb{E} \Bigg[ \frac{\mathbb{I}(A=a) \cdot Y}{p(A=a|S=1,G,\mathbf{\tilde{X}}, \mathbf{\tilde{U}})}  \cdot \frac{p(S=2|G, \mathbf{\Breve{X}}, \mathbf{\Breve{U}}) }{p(S=1|G,\mathbf{\Breve{X}}, \mathbf{\Breve{U}}) } \cdot \frac{p(S=1)}{p(S=2)} \bigg|S=1 \Bigg] \\
    &= \mathbb{E} \Bigg[ \frac{\mathbb{I}(A=a) \cdot Y}{p(A=a|S=1,\mathbf{\tilde{X}}, \mathbf{\tilde{U}})}  \cdot \frac{p(S=2|G, \mathbf{\Breve{X}}, \mathbf{\Breve{U}}) }{p(S=1|G,\mathbf{\Breve{X}}, \mathbf{\Breve{U}}) } \cdot \frac{p(S=1)}{p(S=2)} \bigg|S=1 \Bigg].
\end{align*}
Using the argument above for $\mathbb{E} \big[Y^{(1)}|S=2 \big]$ and $\mathbb{E} \big[Y^{(0)}|S=2 \big]$ completes the proof.
\end{proof}

\begin{proof}[Proof of Theorem \ref{thmbiasdecomposition1}] We will begin by writing
\begin{align*}
    \mathbb{E} \bigg[ \frac{\mathbb{I}(A=a) \cdot Y}{p(A=a|S=1, \mathbf{X}_{AY})} \bigg|S=1\bigg] &= \mathbb{E} \bigg[ \frac{\mathbb{I}(A=a) \cdot Y^{(a)}}{p(A=a|S=1, \mathbf{X}_{AY})} \bigg|S=1\bigg] \\
    &=\mathbb{E} \Bigg[ \mathbb{E} \bigg[\frac{\mathbb{I}(A=a) \cdot Y^{(a)}}{p(A=a|S=1, \mathbf{X}_{AY})}\bigg|S=1, \tilde{\mathbf{X}}, \tilde{\mathbf{U}}, Y^{(a)}\bigg]\Bigg|S=1 \Bigg] \\
    &=\mathbb{E} \bigg[ \frac{Y^{(a)}}{p(A=a|S=1, \mathbf{X}_{AY})} \cdot \mathbb{E} \big[ \mathbb{I}(A=a) \big| S=1, \tilde{\mathbf{X}}, \tilde{\mathbf{U}}, Y^{(a)}\big]\bigg|S=1 \bigg] \\
    &= \mathbb{E} \bigg[ \frac{Y^{(a)}}{p(A=a|S=1, \mathbf{X}_{AY})} \cdot \mathbb{E} \big[ \mathbb{I}(A=a) | S=1, \tilde{\mathbf{X}}, \tilde{\mathbf{U}}\big]\bigg|S=1 \bigg] \\
    &= \mathbb{E} \bigg[Y^{(a)} \cdot  \frac{p(A=a|S=1, \tilde{\mathbf{X}}, \tilde{\mathbf{U}})}{p(A=a|S=1, \mathbf{X}_{AY})} \bigg| S=1  \bigg] = \mathbb{E} \big[Y^{(a)} \cdot  \varepsilon_a \big| S = 1 \big].
\end{align*}
Next, note that since $\mathbb{E}[ \varepsilon_a | S = 1 ] = 1$, we can write:
\begin{align*}
\mathbb{E} \bigg[ \frac{\mathbb{I}(A=a) \cdot Y}{p(A=a|S=1, \mathbf{X}_{AY})} \bigg|S=1\bigg] - \mathbb{E} \big[Y^{(a)} \big|  S=1 \big] &= \mathbb{E} \big[Y^{(a)} \cdot  \varepsilon_a \big| S = 1 \big] - \mathbb{E} \big[Y^{(a)} \big|  S=1 \big] \\
&= \mathbb{E} \bigg[Y^{(a)} \cdot (\varepsilon_a - 1) \bigg|  S=1 \bigg] \\
&= \text{Cov}\big(Y^{(a)}, \varepsilon_a \big| S = 1 \big) \\
&=\text{Cov}\big(Y^{(a,0)}, \varepsilon_a | S = 1  \big) + \text{Cov}\big(\gamma(a), \varepsilon_a \big|S=1 \big)
\end{align*}
Going through the above argument for $A=0$ and $A=1$ proves that $\psi - \phi_1$ can be written as in Theorem 1.
\end{proof}
\begin{proof}[Proof of Theorem \ref{thmbiasdecomposition2}] We have already shown that $\psi - \phi_{1} = T_1 + T_2$. Thus, we only have to show that $\phi_{1} - \phi_{2} = T_3 +T_4$. We let  $\mathscr{X} = \mathbf{\tilde{X}} \cup \mathbf{\breve{X}} $ and $\mathscr{U} = \mathbf{\tilde{U}} \cup \mathbf{\breve{U}}$ and note that by the results of Proposition \ref{unbiasedIPW2}, we have
\begin{center}
\resizebox{1\textwidth}{!}{
\begin{minipage}{\textwidth}
\begin{align*}
    \mathbb{E}[Y^{(a)}|S=2] &=  \mathbb{E} \Bigg[\frac{\mathbb{I}(A=a) \cdot Y}{p(A=a|S=1,\mathbf{\tilde{X}}, \mathbf{\tilde{U}})}  \cdot \frac{p(S=2|G, \mathbf{\Breve{X}}, \mathbf{\Breve{U}}) }{p(S=1|G,\mathbf{\Breve{X}}, \mathbf{\Breve{U}}) } \cdot \frac{p(S=1)}{p(S=2)} \bigg|S=1 \Bigg]  \\
    &= \mathbb{E} \Bigg[\frac{\mathbb{I}(A=a) \cdot Y^{(a)}}{p(A=a|S=1,\mathbf{\tilde{X}}, \mathbf{\tilde{U}})}  \cdot \frac{p(S=2|G, \mathbf{\Breve{X}}, \mathbf{\Breve{U}}) }{p(S=1|G,\mathbf{\Breve{X}}, \mathbf{\Breve{U}}) } \cdot \frac{p(S=1)}{p(S=2)} \bigg|S=1 \Bigg] \\
    &= \mathbb{E} \Bigg[ \mathbb{E}\bigg[\frac{\mathbb{I}(A=a) \cdot Y^{(a)}}{p(A=a|S=1,\mathbf{\tilde{X}}, \mathbf{\tilde{U}})}  \cdot \frac{p(S=2|G, \mathbf{\Breve{X}}, \mathbf{\Breve{U}}) }{p(S=1|G,\mathbf{\Breve{X}}, \mathbf{\Breve{U}}) } \cdot \frac{p(S=1)}{p(S=2)} \bigg|S=1, Y^{(a)}, G,  \mathscr{X}, \mathscr{U} \bigg] \bigg|S=1 \Bigg] \\
    &= \mathbb{E} \Bigg[ \frac{Y^{(a)}}{p(A=a|S=1,\mathbf{\tilde{X}}, \mathbf{\tilde{U}})}  \cdot \frac{p(S=2|G, \mathbf{\Breve{X}}, \mathbf{\Breve{U}}) }{p(S=1|G,\mathbf{\Breve{X}}, \mathbf{\Breve{U}}) } \cdot \frac{p(S=1)}{p(S=2)} \cdot \mathbb{E}\bigg[\mathbb{I}(A=a) \bigg|S=1, Y^{(a)}, G,  \mathscr{X}, \mathscr{U} \bigg] \bigg|S=1 \Bigg] \\
    &= \mathbb{E} \Bigg[ \frac{Y^{(a)}}{p(A=a|S=1,\mathbf{\tilde{X}}, \mathbf{\tilde{U}})}  \cdot \frac{p(S=2|G, \mathbf{\Breve{X}}, \mathbf{\Breve{U}}) }{p(S=1|G,\mathbf{\Breve{X}}, \mathbf{\Breve{U}}) } \cdot \frac{p(S=1)}{p(S=2)} \cdot \mathbb{E}\bigg[\mathbb{I}(A=a) \bigg|S=1, Y^{(a)}, G,  \mathbf{\tilde{X}}, \mathbf{\tilde{U}} \bigg] \bigg|S=1 \Bigg] \\
    &= \mathbb{E} \Bigg[ \frac{Y^{(a)}}{p(A=a|S=1,\mathbf{\tilde{X}}, \mathbf{\tilde{U}})}  \cdot \frac{p(S=2|G, \mathbf{\Breve{X}}, \mathbf{\Breve{U}}) }{p(S=1|G,\mathbf{\Breve{X}}, \mathbf{\Breve{U}}) } \cdot \frac{p(S=1)}{p(S=2)} \cdot \mathbb{E}\bigg[\mathbb{I}(A=a) \bigg|S=1, \mathbf{\tilde{X}}, \mathbf{\tilde{U}} \bigg] \bigg|S=1 \Bigg] \\
    &= \mathbb{E} \Bigg[ Y^{(a)}  \cdot \frac{p(S=2|G, \mathbf{\Breve{X}}, \mathbf{\Breve{U}}) }{p(S=1|G,\mathbf{\Breve{X}}, \mathbf{\Breve{U}}) } \cdot \frac{p(S=1)}{p(S=2)} \bigg|S=1 \Bigg].
\end{align*}
\end{minipage}
}
\end{center}
Therefore, we can write 
\begin{align*}
    \mathbb{E}[Y^{(a)}|S=1] - \mathbb{E}[Y^{(a)}|S=2] &= \mathbb{E} \Bigg[ Y^{(a)}  \cdot \bigg(1 -\frac{p(S=2|G, \mathbf{\Breve{X}}, \mathbf{\Breve{U}}) }{p(S=1|G,\mathbf{\Breve{X}}, \mathbf{\Breve{U}}) } \cdot \frac{p(S=1)}{p(S=2)} \bigg) \bigg|S=1 \Bigg] \\
    &= \textup{Cov}\big(Y^{(a)}, \tilde{\varepsilon}_S \big|S=1 \big) \\
    &= \textup{Cov}\big(Y^{(a,0)}, \tilde{\varepsilon}_S\big| S=1\big)+\textup{Cov}\big(\gamma(a), \tilde{\varepsilon}_S \big|S=1\big).
\end{align*}
We can use the argument above for both values of $A$ to show that $\phi_1 - \phi_2 = T_3 + T_4$, completing the proof.
\end{proof}
\begin{proof}[Proof of Proposition \ref{propositionoutcomegeneratingequation}] To prove Proposition \ref{propositionoutcomegeneratingequation}, we will develop a different bias decomposition for $\phi_1 - \phi_2$. We note that
\begin{align*}
    \mathbb{E}[Y^{(1)} - Y^{(0)}|S=s] &= \mathbb{E}\bigg[\frac{\mathbb{I}(S=s) \cdot (Y^{(1)} - Y^{(0)})}{p(S=s)}\bigg] \\
    &= \sum_{g=0}^{g_{\text{max}}} \mathbb{E} \bigg[\frac{\mathbb{I}(S=s, G=g) \cdot \tau(g)}{p(S=s)}\bigg].
\end{align*}
Next, since we assumed that $Y$ and $S$ share no common causes, and $S$ only has a direct causal effect on the personal and neighborhood treatments, we have the following two relationships: (i) $S, G \ind \tau(g) | \mathbf{X}_{GY}, \mathbf{U}_{GY}$, and (ii) $S \ind \mathbf{X}_{GY}, \mathbf{U}_{GY}$. Thus, we can write:

\begin{align*}
    \mathbb{E}[Y^{(1)} - Y^{(0)}|S=s] &= \sum_{g=0}^{g_{\text{max}}} \mathbb{E} \Bigg[ \mathbb{E} \bigg[ \frac{\mathbb{I}(S=s, G=g) \cdot \tau(g)}{p(S=s)} \Bigg| \tau(g), \mathbf{X}_{GY}, \mathbf{U}_{GY}\bigg] \Bigg] \\
    &= \sum_{g=0}^{g_{\text{max}}} \mathbb{E} \Bigg[ \frac{\tau(g)}{p(S=s)} \cdot \mathbb{E}  \bigg[ \mathbb{I}(S=s, G=g) \bigg| \tau(g), \mathbf{X}_{GY}, \mathbf{U}_{GY}\bigg] \Bigg] \\
    &= \sum_{g=0}^{g_{\text{max}}} \mathbb{E} \Bigg[ \frac{\tau(g)}{p(S=s)} \cdot \mathbb{E}  \bigg[ \mathbb{I}(S=s, G=g) \bigg| \mathbf{X}_{GY}, \mathbf{U}_{GY}\bigg] \Bigg] \\
    &= \sum_{g=0}^{g_{\text{max}}} \mathbb{E} \Bigg[ \frac{\tau(g)}{p(S=s)} \cdot p(S=s, G=g| \mathbf{X}_{GY}, \mathbf{U}_{GY}) \Bigg] \\
    &= \sum_{g=0}^{g_{\text{max}}} \mathbb{E} \Bigg[ \frac{\tau(g)}{p(S=s)} \cdot p(G=g|S=s, \mathbf{X}_{GY}, \mathbf{U}_{GY}) \cdot p(S=s |\mathbf{X}_{GY}, \mathbf{U}_{GY}) \Bigg] \\
    &= \sum_{g=0}^{g_{\text{max}}} \mathbb{E} \Bigg[ \tau(g)\cdot p(G=g|S=s, \mathbf{X}_{GY}, \mathbf{U}_{GY})\Bigg] \\
    &= \sum_{g=0}^{g_{\text{max}}} \text{Cov} \big( \tau(g), p(G=g|S=s, \mathbf{X}_{GY}, \mathbf{U}_{GY}) \big) \\
    &+  \sum_{g=0}^{g_{\text{max}}} \mathbb{E}\big[ \tau(g) \big] \cdot \mathbb{E} \big[p(G=g|S=s, \mathbf{X}_{GY}, \mathbf{U}_{GY}) \big] \\
    &= \sum_{g=0}^{g_{\text{max}}} \text{Cov} \big( \tau(g), p(G=g|S=s, \mathbf{X}_{GY}, \mathbf{U}_{GY}) \big) \\
    &+  \sum_{g=0}^{g_{\text{max}}} \mathbb{E}\big[ \tau(g) \big] \cdot p(G=g|S=s) 
\end{align*}
Utilizing the above argument for $S=1$ and $S=2$, we can write 
\begin{align*}
    \phi_1 -\phi_2 &= \sum_{g=0}^{g_{\text{max}}} \text{Cov} \big( \tau(g), p(G=g|S=1, \mathbf{X}_{GY}, \mathbf{U}_{GY}) -  p(G=g|S=2, \mathbf{X}_{GY}, \mathbf{U}_{GY})\big) \\
    &+  \sum_{g=0}^{g_{\text{max}}} \mathbb{E}\big[ \tau(g) \big] \cdot \big( p(G=g|S=1) - p(G=g|S=2) \big).
\end{align*}

Next, we note that since we assumed that $G$ is discrete, for $g \in \{1,2, ..., g_{\text{max}}\}$ and $s \in \{1,2\}$, we have $p(G=g|S=s) = p(G \leq g | S=s) - p(G \leq g-1 |S =s)$ and $p(G=g|S=s, \mathbf{X}_{GY}, \mathbf{U}_{GY}) = p(G \leq g | S=s, \mathbf{X}_{GY}, \mathbf{U}_{GY}) - p(G \leq g-1 |S =s, \mathbf{X}_{GY}, \mathbf{U}_{GY})$. Thus, for $s \in \{1,2\}$,
{\footnotesize
\begin{align*}
    \sum_{g=0}^{g_{\text{max}}} \text{Cov}\bigg(\tau(g), p(G=g|S=s, \mathbf{X}_{GY}, \mathbf{U}_{GY})\bigg) = \sum_{g=1}^{g_{\text{max}}} \text{Cov}\bigg(\tau(g-1) - \tau(g), p(G\leq g-1|S=s, \mathbf{X}_{GY}, \mathbf{U}_{GY})\bigg).
\end{align*}
}
Similarly, for $s \in \{1,2\}$, we have 
{\small
\begin{align*}
    \sum_{g=0}^{g_{\text{max}}}\mathbb{E}[\tau(g)]\cdot p(G=g|S=s) &= \sum_{g=1}^{g_{\text{max}}} \bigg\{ \mathbb{E}\big[\tau(g-1) - \tau(g)\big] \cdot p(G \leq g-1| S=s) \bigg\} + \mathbb{E}\big[\tau(g_{max})\big].
\end{align*}
}
Finally, if $Y= f_0(\mathbf{X}, \mathbf{U}) + f_1(A, \mathbf{X}, \mathbf{U}) + f_2(G, \mathbf{X}, \mathbf{U}) + \epsilon(A,G)$, for any $g\in\{1, 2, ..., g_{\text{max}}\}$, $\text{Cov}\big(\tau(g-1) - \tau(g), p(G\leq g-1|S=s, \mathbf{X}_{GY}, \mathbf{U}_{GY})\big) = 0$ and $ \mathbb{E}\big[\tau(g-1) - \tau(g)\big] = 0$. Thus, $\phi_1 - \phi_2 = \mathbb{E}\big[\tau(g_{max})\big] - \mathbb{E}\big[\tau(g_{max})\big] = 0$, completing the proof.
\end{proof}

\begin{proof}[Proof of Proposition \ref{unbiasedIPW1}] We begin by proving the identification formula for $\phi_1$. We note that
    \begin{align*}
    \mathbb{E} \big[Y^{(a)}|S=1, N=n \big] &= \mathbb{E} \Bigl[ \mathbb{E} \big[Y^{(a)}\big|S=1, N=n, \tilde{\mathbf{X}}, \tilde{\mathbf{U}}\big] \Bigl|S=1, N=n\Bigl] \\
    &= \mathbb{E} \Bigl[ \mathbb{E} \big[Y^{(a)}\big|S=1, N=n, A=a, \tilde{\mathbf{X}}, \tilde{\mathbf{U}}\big]\Bigl|S=1, N=n \Bigl] \\
    &= \mathbb{E} \Bigl[ \mathbb{E} \big[Y\big|S=1, N=n, A=a, \tilde{\mathbf{X}}, \tilde{\mathbf{U}}\big]\Bigl|S=1, N=n \Bigl] \\
    &= \int_{\tilde{\mathbf{x}}, \tilde{\mathbf{u}}, y} y \cdot p(y|S=1, N=n, A=a,\tilde{\mathbf{x}}, \tilde{\mathbf{u}}) \cdot p(\tilde{\mathbf{x}}, \tilde{\mathbf{u}} | S=1, N=n) \\
    &= \int_{\tilde{\mathbf{x}}, \tilde{\mathbf{u}}, y} y \cdot \frac{p(y, A=a, \tilde{\mathbf{x}}, \tilde{\mathbf{u}}|S=1, N=n)}{p(A=a|S=1, N=n, \tilde{\mathbf{x}}, \tilde{\mathbf{u}})} \\
    &= \int_{\tilde{\mathbf{x}}, \tilde{\mathbf{u}}, y, \tilde{a}} \mathbb{I}(\tilde{a}=a) \cdot y \cdot \frac{p(y, \tilde{a}, \tilde{\mathbf{x}}, \tilde{\mathbf{u}}|S=1, N=n)}{p(A=a|S=1, N=n, \tilde{\mathbf{x}}, \tilde{\mathbf{u}})} \\
    &= \mathbb{E} \bigg[ \frac{\mathbb{I}(A=a) \cdot Y}{p(A=a|S=1, N=n, \tilde{\mathbf{X}}, \tilde{\mathbf{U}})} \bigg| S=1, N=n \bigg] \\
    &= \mathbb{E} \bigg[ \frac{\mathbb{I}(A=a) \cdot Y}{p(A=a|S=1, \tilde{\mathbf{X}}, \tilde{\mathbf{U}})} \bigg| S=1, N=n \bigg].
\end{align*}
The derivation above proves the first half of Proposition \ref{unbiasedIPW1} by the law of total expectation. Next, we prove our identification result for $\phi_2$. We can write 
\begin{center}
\resizebox{0.9\textwidth}{!}{
\begin{minipage}{\textwidth}
\begin{align*}
    \mathbb{E} \big[Y^{(a)}|S=2, N=n \big] &= \mathbb{E} \Bigl[ \mathbb{E} \big[Y^{(a)}\big|S=2, N=n, G, \mathfrak{X}, \mathfrak{U}\big] \Bigl|S=2, N=n \Bigl] \\
    &= \mathbb{E} \Bigl[ \mathbb{E} \big[Y^{(a)}\big|S=1, N=n, G, \mathfrak{X}, \mathfrak{U}\big] \Bigl|S=2,N=n \Bigl] \\
    &= \mathbb{E} \Bigl[ \mathbb{E} \big[Y^{(a)}\big|S=1 ,N=n, A=a,G, \mathfrak{X}, \mathfrak{U}\big] \Bigl|S=2, N=n \Bigl] \\
    &= \int_{g,\mathfrak{x},\mathfrak{u},y} y \cdot p(y|S=1, N=n, A=a, g, \mathfrak{x}, \mathfrak{u}) \cdot p(g,\mathfrak{x},\mathfrak{u}|S=2, N=n) \\
    &= \int_{g,\mathfrak{x},\mathfrak{u},y} y \cdot \frac{p(y, A=a, g, \mathfrak{x}, \mathfrak{u}|S=1, N=n)}{p(A=a|S=1,N=n, g, \mathfrak{x},\mathfrak{u})} \cdot \frac{p(g, \mathfrak{x},\mathfrak{u}|S=2,N=n)}{p(g, \mathfrak{x},\mathfrak{u}|S=1,N=n)} \\ 
    &= \int_{g,\mathfrak{x},\mathfrak{u},y} y \cdot \frac{p(y, A=a, g, \mathfrak{x}, \mathfrak{u}|S=1, N=n)}{p(A=a|S=1, N=n, g, \mathfrak{x},\mathfrak{u})} \cdot \frac{p(S=2 |N=n, g, \mathfrak{x},\mathfrak{u})}{p(S=1| N=n, g, \mathfrak{x},\mathfrak{u})} \cdot \frac{p(S=1|N=n)}{p(S=2| N=n)} \\
    &= \int_{g,\mathfrak{x},\mathfrak{u},y, \tilde{a}} \mathbb{I}(\tilde{a}=a) \cdot y \cdot \frac{p(y, \tilde{a}, g, \mathfrak{x}, \mathfrak{u}|S=1, N=n)}{p(A=a|S=1, N=n, g, \mathfrak{x},\mathfrak{u})} \cdot \frac{p(S=2|N=n, g, \mathfrak{x},\mathfrak{u})}{p(S=1|N=n,g, \mathfrak{x},\mathfrak{u})} \cdot \frac{p(S=1|N=n)}{p(S=2|N=n)} \\
    &= \mathbb{E} \Bigg[ \frac{\mathbb{I}(A=a) \cdot Y}{p(A=a|S=1,N=n, G,\mathfrak{X}, \mathfrak{U})}  \cdot \frac{p(S=2|N=n, G, \mathfrak{X}, \mathfrak{U}) }{p(S=1|N=n, G, \mathfrak{X}, \mathfrak{U}) } \cdot \frac{p(S=1|N=n)}{p(S=2|N=n)} \bigg|S=1, N=n \Bigg] \\
    &= \mathbb{E} \Bigg[ \frac{\mathbb{I}(A=a) \cdot Y}{p(A=a|S=1,N=n,\mathbf{\tilde{X}}, \mathbf{\tilde{U}})}  \cdot \frac{p(S=2|N=n,G, \mathbf{\Breve{X}}, \mathbf{\Breve{U}}) }{p(S=1|N=n, G,\mathbf{\Breve{X}}, \mathbf{\Breve{U}}) } \cdot \frac{p(S=1|N=n)}{p(S=2|N=n)} \bigg|S=1, N=n \Bigg] \\
    &= \mathbb{E} \Bigg[ \frac{\mathbb{I}(A=a) \cdot Y}{p(A=a|S=1, \mathbf{\tilde{X}}, \mathbf{\tilde{U}})}  \cdot \frac{p(S=2|G, \mathbf{\Breve{X}}, \mathbf{\Breve{U}}) }{p(S=1| G,\mathbf{\Breve{X}}, \mathbf{\Breve{U}}) } \cdot \frac{p(S=1)}{p(S=2)} \bigg|S=1, N=n \Bigg].
    \end{align*}
\end{minipage}
}
\end{center}
The argument above can be used for $A=0$ and $A=1$, which, when combined with the law of total expectation, completes the proof of the identification formula for $\phi_2$.
\end{proof}

\begin{proof}[Proof of Theorem \ref{thmupdatedbiasdecomposition}]
 We will begin by proving the formula of the bias of $\hat{\psi}$ for $\phi_1$. We note that
 \begin{center}
\resizebox{0.9\textwidth}{!}{
\begin{minipage}{\textwidth}
 \begin{align*}
    \mathbb{E} \bigg[ \frac{\mathbb{I}(A=a) \cdot Y}{p(A=a|S=1, \mathbf{X}_{AY})} \bigg|S=1\bigg] &= \sum_{n=0}^{g_{\textup{max}}} \mathbb{E} \bigg[ \frac{\mathbb{I}(A=a) \cdot Y}{p(A=a|S=1, \mathbf{X}_{AY})} \bigg|S=1, N = n\bigg] \cdot p(N=n|S=1) \\
    &= \sum_{n=0}^{g_{\textup{max}}} \mathbb{E} \bigg[ \frac{\mathbb{I}(A=a) \cdot Y^{(a)}}{p(A=a|S=1, \mathbf{X}_{AY})} \bigg|S=1, N = n\bigg] \cdot p(N=n|S=1) \\
    &=\sum_{n=0}^{g_{\textup{max}}}\mathbb{E} \Bigg[ \mathbb{E} \bigg[\frac{\mathbb{I}(A=a) \cdot Y^{(a)}}{p(A=a|S=1, \mathbf{X}_{AY})}\bigg|S=1, N=n, \tilde{\mathbf{X}}, \tilde{\mathbf{U}}, Y^{(a)}\bigg]\Bigg|S=1, N=n \Bigg] \cdot p(N=n|S=1) \\
    &= \sum_{n=0}^{g_{\textup{max}}} \mathbb{E} \bigg[ \frac{Y^{(a)}}{p(A=a|S=1, \mathbf{X}_{AY})} \cdot \mathbb{E} \big[ \mathbb{I}(A=a) \big| S=1, N=n, \tilde{\mathbf{X}}, \tilde{\mathbf{U}}, Y^{(a)}\big]\bigg|S=1, N=n \bigg] \cdot p(N=n|S=1) \\
    &= \sum_{n=0}^{g_{\textup{max}}} \mathbb{E} \bigg[Y^{(a)} \cdot  \frac{p(A=a|S=1, N=n, \tilde{\mathbf{X}}, \tilde{\mathbf{U}})}{p(A=a|S=1, \mathbf{X}_{AY})} \bigg| S=1, N=n  \bigg] \cdot p(N=n|S=1) \\
    &= \sum_{n=0}^{g_{\textup{max}}} \mathbb{E} \bigg[Y^{(a)} \cdot  \frac{p(A=a|S=1, \tilde{\mathbf{X}}, \tilde{\mathbf{U}})}{p(A=a|S=1, \mathbf{X}_{AY})} \bigg| S=1, N=n  \bigg] \cdot p(N=n|S=1).
\end{align*}
\end{minipage}
}
\end{center}
Moreover, we can write
\begin{center}
\resizebox{0.9\textwidth}{!}{
\begin{minipage}{\textwidth}
\begin{align*}
\mathbb{E} \bigg[ \frac{\mathbb{I}(A=a) \cdot Y}{p(A=a|S=1, \mathbf{X}_{AY})} \bigg|S=1, N=n\bigg] - \mathbb{E} \big[Y^{(a)} \big|  S=1, N=n \big] &= \mathbb{E} \big[Y^{(a)} \cdot  \varepsilon_a \big| S = 1, N=n \big] - \mathbb{E} \big[Y^{(a)} \big|  S=1, N=n \big] \\
&= \mathbb{E} \bigg[Y^{(a)} \cdot (\varepsilon_a - 1) \bigg|  S=1, N=n \bigg] \\
&= \text{Cov}\big(Y^{(a)}, \varepsilon_a \big| S = 1, N=n \big) \\
&=\text{Cov}\big(Y^{(a,0)}, \varepsilon_a | S = 1, N=n  \big) \\
&+ \text{Cov}\big(\gamma(a), \varepsilon_a \big|S=1, N=n \big).
\end{align*}
\end{minipage}
}
\end{center}
Using the fact that $\sum\limits_{n=0}^{g_{\text{max}}} p(N=n|S=1) \cdot \text{Cov}\big(Y^{(a,0)}, \varepsilon_a | S = 1, N=n  \big) = \text{Cov}\big(Y^{(a,0)}, \varepsilon_a | S = 1\big)$ and combining the two results above for both values of $A$ proves the first half of Theorem \ref{thmupdatedbiasdecomposition}. Next, we prove the second part of Theorem \ref{thmupdatedbiasdecomposition}, which represents the difference between $\phi_1$ and $\phi_2$. We note that by Proposition \ref{unbiasedIPW1},
\begin{center}
\resizebox{0.9\textwidth}{!}{
\begin{minipage}{\textwidth}
\begin{align*}
    \mathbb{E} \big[Y^{(a)}|S=2, N=n \big] &= \mathbb{E} \Bigg[ \frac{\mathbb{I}(A=a) \cdot Y^{(a)}}{p(A=a|S=1,\mathbf{\tilde{X}}, \mathbf{\tilde{U}})}  \cdot \frac{p(S=2|G, \mathbf{\Breve{X}}, \mathbf{\Breve{U}}) }{p(S=1| G,\mathbf{\Breve{X}}, \mathbf{\Breve{U}}) } \cdot \frac{p(S=1)}{p(S=2)} \bigg|S=1, N=n \Bigg] \\
    &= \mathbb{E} \Bigg[ \mathbb{E} \bigg[ \frac{\mathbb{I}(A=a) \cdot Y^{(a)}}{p(A=a|S=1,\mathbf{\tilde{X}}, \mathbf{\tilde{U}})}  \cdot \frac{p(S=2|G, \mathbf{\Breve{X}}, \mathbf{\Breve{U}}) }{p(S=1|G,\mathbf{\Breve{X}}, \mathbf{\Breve{U}}) } \cdot \frac{p(S=1)}{p(S=2)}\bigg| S=1, N=n, Y^{(a)}, G,  \mathscr{X}, \mathscr{U}\bigg] \bigg|S=1, N=n \Bigg] \\
    &=  \mathbb{E} \Bigg[ Y^{(a)}  \cdot \frac{p(S=2|G, \mathbf{\Breve{X}}, \mathbf{\Breve{U}}) }{p(S=1| G,\mathbf{\Breve{X}}, \mathbf{\Breve{U}}) } \cdot \frac{p(S=1)}{p(S=2)} \bigg|S=1, N=n \Bigg].
\end{align*}
\end{minipage}
}
\end{center}
Thus, we have that
\begin{align*}
    \phi_{2} = \sum_{n=0}^{g_{\text{max}}} &\mathbb{E} \Bigg[ \kappa \cdot \frac{p(S=2| G, \mathbf{\breve{X}}, \mathbf{\breve{U}}) }{p(S=1|G, \mathbf{\breve{X}}, \mathbf{\breve{U}}) } \cdot \frac{p(S=1)}{p(S=2)} \bigg| S=1, N=n \Bigg] \cdot p(N=n|S=1).
\end{align*}
Next, we note that 
\begin{center}
\resizebox{1\textwidth}{!}{
\begin{minipage}{\textwidth}
\begin{align*}
    \mathbb{E}\bigg[\kappa\bigg|S=1, N=n\bigg] - \mathbb{E}\bigg[\kappa \cdot \frac{p(S=2| G, \mathbf{\breve{X}}, \mathbf{\breve{U}}) }{p(S=1|G, \mathbf{\breve{X}}, \mathbf{\breve{U}}) } \cdot \frac{p(S=1)}{p(S=2)}\bigg|S=1, N=n\bigg] &= \mathbb{E} \Bigg[ \kappa  \cdot \bigg(1 -\frac{p(S=2|G, \mathbf{\Breve{X}}, \mathbf{\Breve{U}}) }{p(S=1|G,\mathbf{\Breve{X}}, \mathbf{\Breve{U}}) } \cdot \frac{p(S=1)}{p(S=2)} \bigg) \bigg|S=1, N=n \Bigg] \\
    &= \textup{Cov}\big(\kappa, \tilde{\varepsilon}_S \big|S=1, N=n \big)  \\
    &= \textup{Cov}\big(\tau(0), \tilde{\varepsilon}_S\big| S=1, N=n\big)+\textup{Cov}\big(\gamma(1) - \gamma(0), \tilde{\varepsilon}_S \big|S=1, N=n\big), 
\end{align*}
\end{minipage}
}
\end{center}
which, when combined with the fact that $\sum\limits_{n=0}^{g_{\text{max}}} p(N=n|S=1) \cdot \textup{Cov}(\tau(0), \tilde{\varepsilon}_S|S=1, N=n) = \text{Cov}(\tau(0), \tilde{\varepsilon}_S | S = 1)$, completes the proof.
\end{proof}

\end{document}